\documentclass[12pt]{article}
\usepackage{amssymb}
\usepackage{amsfonts}
\usepackage{amsmath}

\newtheorem{theorem}{Theorem}

\newtheorem{claim}{Claim}

\newtheorem{corollary}{Corollary}

\newtheorem{definition}{Definition}

\newtheorem{lemma}{Lemma}

\newtheorem{proposition}{Proposition}
\newtheorem{remark}{Remark}

\begin{document}

\title{School Choice Problems with Priority Adjustments\thanks{%
This paper builds on some of the results contained in an earlier manuscript
entitled \textquotedblleft School Choice with Multiple
Priorities,\textquotedblright\ while adding new results and substantially
revising. This work was supported by JSPS KAKENHI Grant Number 25K05004.}}
\author{Minoru Kitahara\thanks{%
Department of Economics, Osaka Metropolitan University} \and Yasunori Okumura%
\thanks{%
Corresponding author. Department of Logistics and Information Engineering,
TUMSAT, 2-1-6, Etchujima, Koto-ku, Tokyo, 135-8533 Japan.
Phone:+81-3-5245-7300. Fax:+81-3-5245-7300. E-mail: yokumu0@kaiyodai.ac.jp}}
\maketitle

\begin{center}
\textbf{Abstract }
\end{center}

This paper studies school choice when each school has an original priority
order and an adjusted priority order reflecting policies such as affirmative
action. We introduce a weak notion of stability that permits justified envy
under one priority order only when the opposite priority relation is
supported by the other. We characterize weakly stable matchings as stable
matchings under the intersection of the two priority profiles. Using the
efficiency-adjusted deferred acceptance algorithm, we obtain a weakly stable
matching that cannot be Pareto improved upon, within the set of weakly
stable matchings, for any group containing all improved students by the
adjustment. We also show that any Pareto improvement over this outcome
necessarily creates justified envy under both priority profiles. Finally,
our mechanism is responsive to stronger priority adjustments: the outcome
under a stronger affirmative action policy is not Pareto dominated by that
under a weaker policy for any group containing all students improved by the
stronger adjustment.


\newpage

\section{Introduction}

This study considers a setting in which each school has both an original
priority order and an adjusted priority order over students. For example, a
school's original priority order may be determined by students' entrance
examination scores, while its adjusted priority order may incorporate policy
considerations such as affirmative action. Thus, if a minority student $i$
has a lower examination score at a school than a majority student $j$, the
original priority order ranks $j$ above $i$, but the adjustment may reverse
their relative priority and rank $i$ above $j$.

Our main insight is that priority relations supported by both the original
and adjusted priority orders should be distinguished from those supported by
only one of them. We regard a priority relation supported by both orders as
a compelling claim that should not be violated. By contrast, when the two
orders rank a pair of students differently, assigning either student is
supported by one of the two externally specified priority criteria.
Accordingly, to accommodate the coexistence of the two priority orders, we
weaken the conventional notion of fairness based on a single priority order:
we allow justified envy under one priority order only when the opposite
priority relation is supported by the other. This weakening permits a
broader set of matchings and thereby creates scope for efficiency
improvements.\footnote{%
A related idea appears in Ayg\"{u}n and B\'{o} (2021), whose fairness
criterion permits a lower-scoring student to be admitted over a
higher-scoring student only when the former has a stronger privilege claim.
Their criterion is formulated in terms of categorical privilege claims,
whereas our framework also accommodates more moderate forms of preferential
treatment, such as bonus-point rules, under which possessing a privileged
status does not necessarily reverse every score-based priority comparison.}

We make three main contributions. First, we introduce and characterize a
notion of weak stability based on the two priority orders. Second, we
propose a mechanism that selects a weakly stable outcome that is not Pareto
dominated by any other weakly stable outcome for the students whose
priorities are improved. Third, we show that this mechanism is responsive to
stronger priority adjustments, although such responsiveness is incompatible
with conventional stability.

First, our study contributes to the literature on the trade-off between
efficiency and stability in school choice; Kitahara and Okumura (2023)
provide an overview of various studies on methods for improving efficiency
while retaining some degree of fairness. In our approach, permissible
priority violations are determined by the schools' exogenously given
original and adjusted priority orders. We characterize weakly stable
matchings under the two priority orders as stable matchings under their
intersection. Based on this characterization and the efficiency-adjusted
deferred acceptance (EADA) algorithm of Kesten (2010), we provide a method
for obtaining a student-optimally weakly stable matching (SOWSM), defined as
a weakly stable matching that is not Pareto dominated by any other weakly
stable matching.\footnote{%
Cerrone et al. (2024) provide an overview of various studies that employ the
EADA algorithm.} We further show that any Pareto improvement over an SOWSM
necessarily creates justified envy supported by both the original and
adjusted priority orders. Thus, such an improvement requires violating a
priority relation that can be justified by neither priority criterion.

Second, we strengthen this efficiency conclusion for students whose
priorities are improved by the adjustment. We call a student an improved
student if, at some school, the student's adjusted priority is higher than
that of another student who had higher priority under the original order. We
show that the EADA algorithm produces a weakly stable matching that is
Pareto undominated, among weakly stable matchings, for any group containing
all improved students. This group-specific result is stronger than Pareto
undominance for all students: it rules out a weakly stable matching that
makes all members of the relevant group weakly better off and at least one
of them strictly better off, even if students outside the group are made
worse off. In particular, if students not targeted by an affirmative action
policy receive no bonus points, then the targeted group contains all
improved students. Hence, no other weakly stable matching Pareto dominates
the EADA outcome with respect to the targeted group.

Third, we establish responsiveness to stronger priority adjustments. We
compare two adjusted priority profiles, one of which induces at least all
the priority improvements induced by the other relative to the original
profile. Although Kojima (2012) shows that no stable mechanism can guarantee
responsiveness to stronger affirmative action, we show that the outcome
under the weaker adjustment does not Pareto dominate the outcome under the
stronger adjustment for any group containing all students whose priorities
are improved under the stronger adjustment. Thus, by requiring weak rather
than conventional stability, our mechanism reconciles meaningful fairness
protection with responsiveness to stronger affirmative action.

Our first application is priority-based affirmative action, which has been
extensively studied in the school choice literature.\footnote{%
The responsiveness to improvements is related to the respect for
improvements property, as introduced by Balinsky and S\"{o}nmez (1999).
However, the latter property considers the improvement for one student only.
Hirata et al. (2022) consider the respect for group improvements property,
which is more related to ours. They show that any stable mechanism fails to
satisfy the property.} See, for example, Abdulkadiro\u{g}lu and Grigoryan
(2025) for an overview of various studies on this topic. Among them, we
consider a priority-based affirmative action policy discussed by Kojima
(2012), Afacan and Salman (2016), Jiao and Tian (2018), Jiao and Shen
(2021), Jiao et al. (2022) and Dur and Xie (2023).\footnote{%
For other types of affirmative action, there are quota-based affirmative
action and reserve-based affirmative action. Moreover, Kitahara and Okumura
(2020, 2026a) propose a hybrid approach that combines these three types of
affirmative action.} For example, in Chinese college admissions, minority
students are awarded bonus points in the national college entrance
examination. See, for example, Wang (2009) for details of this policy.

Jiao and Shen (2021) obtain responsiveness by allowing only minority
students to consent under EADA, but the resulting matching may violate the
priorities among minority students. By contrast, our mechanism permits only
violations justified by disagreements between the original and adjusted
priorities. Unlike the approaches of Jiao et al. (2022) and Dur and Xie
(2023), we impose no additional restriction on priority profiles; instead,
we weaken stability while retaining stronger welfare protection for the
targeted students.

Our second application concerns a setting in which each school uses two
component tests to determine its priority order; Kitahara and Okumura (2021)
briefly discuss such a setting. Each school may choose the weight assigned
to one component from an admissible interval and rank students according to
the resulting weighted-average scores. For each admissible weight profile,
we construct a matching that is weakly stable with respect to the priority
orders generated by the endpoints of these intervals. Then, this matching
weakly Pareto-dominates every matching that is stable under the resulting
priority profile. If it admits justified envy under that profile, the Pareto
dominance is strict. Nevertheless, every such instance of justified envy is
reversed under some other admissible choice of weights: under that choice,
the envying student has a lower composite score than the student she envies.

\section{Model}

\subsection{Basic Setting}

Let $I$ and $S$ be the sets of students and schools, respectively, where $%
\left\vert I\right\vert \geq 3$ and $\left\vert S\right\vert \geq 3$. Let $%
q_{s}$ be the capacity of schools $s\in S$. Let $q=\left( q_{s}\right)
_{s\in S}$ denote the capacity vector.

Each student $i\in I$ has a (strict) linear order (asymmetric, transitive,
and complete) on $S\cup \left\{ \emptyset \right\} $ denoted by $P_{i}$
representing the preferences of student $i$, where $sP_{i}s^{\prime }$ (or
equivalently $\left( s,s^{\prime }\right) \in P_{i}$) means that $i$ prefers 
$s\in S\cup \left\{ \emptyset \right\} $ to $s^{\prime }\in S\cup \left\{
\emptyset \right\} $ and $\emptyset $ represents the outside option. If $%
sP_{i}\emptyset $, then school $s$ is said to be \textbf{acceptable} to $i$.
Further, $sR_{i}s^{\prime }$ means $sP_{i}s^{\prime }$ or $s=s^{\prime }$.

Let $\succ _{s}$ denote the priority order of school $s$, where $i\succ
_{s}j $, or equivalently $(i,j)\in \succ _{s}$, means that student $i$ has
higher priority than student $j$ at school $s$. We assume that $\succ _{s}$
is a strict partial order, that is, transitive and asymmetric, for every $%
s\in S$. See Kitahara and Okumura (2026b) for a study of school-choice
problems with partial priority orders.

We say that a \textbf{matching} $\mu $ is a mapping satisfying $\mu (i)\in
S\cup \{\emptyset \},$ $\mu \left( s\right) \subseteq I$, $\mu (i)=s$ if and
only if $i\in \mu \left( s\right) ,$ and $\left\vert \mu \left( s\right)
\right\vert \leq q_{s}$. Note that $\mu (i)=\emptyset $ means that $i$ is
unmatched to any school and $\mu (i)=s\in S$ means that $i$ is matched to $s$
under a matching $\mu $. A matching $\mu $ is said to be \textbf{%
individually rational} if $\mu \left( i\right) R_{i}\emptyset $ for all $%
i\in I$. A matching $\mu $ is said to be \textbf{non-wasteful} if $sP_{i}\mu
\left( i\right) $ implies $\left\vert \mu \left( s\right) \right\vert =q_{s}$
for all $i\in I$ and all $s\in S$.

Student $i$ has \textbf{justified envy} under $\succ $ toward $j$ at $\mu $
if $s=\mu \left( j\right) P_{i}\mu \left( i\right) \ $and $i\succ _{s}j$. If
no student has justified envy under $\succ $ at $\mu $, then $\mu $ is said
to be \textbf{fair }under $\succ $. A matching $\mu $ is \textbf{stable}
under $\succ $ if it is individually rational, non-wasteful and fair under $%
\succ $.

A matching $\mu $ is \textbf{Pareto dominated for }$I^{\prime }\subseteq I$
by $\mu ^{\prime }$ if $\mu ^{\prime }\left( i\right) P_{i}\mu \left(
i\right) $\ for some $i\in I^{\prime }$ and $\mu ^{\prime }\left( i^{\prime
}\right) R_{i^{\prime }}\mu \left( i^{\prime }\right) $ for all $i^{\prime
}\in I^{\prime }$. Moreover, a matching $\mu $ is\textbf{\ weakly} \textbf{%
Pareto dominated for }$I^{\prime }\subseteq I$ by $\mu ^{\prime }$ if it is
Pareto dominated for $I^{\prime }$ by $\mu ^{\prime }$ or $\mu \left(
i\right) =\mu ^{\prime }\left( i\right) $ for all $i\in I^{\prime }$. When $%
I^{\prime }=I$ is additionally satisfied, we simply write that $\mu $ 
\textbf{is (weakly) Pareto dominated }by $\mu ^{\prime }$, as defined in
many previous studies.

A matching $\mu $ is an $I^{\prime }$\textbf{-optimally stable matching}
under $\succ $\textbf{\ }if it is stable under $\succ $ and is not Pareto
dominated for $I^{\prime }$ by any stable matching under $\succ $. When $%
I^{\prime }=I$ is additionally satisfied, a matching $\mu $ is said to be 
\textbf{student-optimally stable matching} (hereafter \textbf{SOSM}) under $%
\succ $\textbf{\ }if it is stable under $\succ $ and is not Pareto dominated
for all students by any stable matching under $\succ $.

\subsection{Two Priority Orders and Weak Fairness}

Each school $s\in S$ has two priority orders on $I$ denoted by $\succ
_{s}^{O}$ and $\succ _{s}^{A}$. We assume that the two priority orders of
each school are linear orders and let $\mathcal{L}$ be the set of all
possible linear orders on $I$. We call $\succ _{s}^{O}$ and $\succ _{s}^{A}$
the \textbf{original priority order} and the \textbf{adjusted priority order}%
, respectively. Moreover, let $\succ ^{O}=\left( \succ _{s}^{O}\right)
_{s\in S}\in \mathcal{L}^{\left\vert S\right\vert }$ be the profile of the
original priority orders, and let $\succ ^{A}=\left( \succ _{s}^{A}\right)
_{s\in S}\in \mathcal{L}^{\left\vert S\right\vert }$ be the profile of the
adjusted priority orders.

That is, $i\succ _{s}^{O}j$ (resp. $i\succ _{s}^{A}j$) means that, under the
original priority (resp. the adjusted priority), student $i$ has higher
priority than student $j$ at school $s$. Note that $\succ _{s}^{O}=\succ
_{s}^{A}$ is allowed to be satisfied for some $s\in S$; that is, some school
is allowed to have only one priority order. However, we assume $\succ
_{s}^{O}\neq \succ _{s}^{A}$ for some $s\in S$; equivalently, $i\succ
_{s}^{O}j$ and $j\succ _{s}^{A}i$ are satisfied for some $i,j\in I$ and some 
$s\in S$.

For example, the original priority order may be based on students' entrance
examination scores, whereas the adjusted priority order may incorporate
affirmative action or other policy adjustments that give additional weight
to students from disadvantaged or underrepresented groups. In this case, the
adjusted priority order may reflect broader policy objectives in addition to
students' academic performance. For instance, $i\succ _{s}^{O}j$ and $j\succ
_{s}^{A}i$ may represent a situation in which student $i$ has a higher
entrance examination score than student $j$, but $j$ has higher priority
under the adjusted priority at $s$ because $j$ belongs to a minority group
whereas $i$ belongs to a majority group. In Section 4, we explicitly discuss
such priority reversals.

\begin{definition}
Suppose that a student $i$ has justified envy toward $j$ at $\mu $ under $%
\succ ^{A}$. We say that the justified envy is\textbf{\ induced by the
adjustment} if $j\succ _{\mu \left( j\right) }^{O}i$.
\end{definition}

This refers to a situation in which the priority comparison that gives rise
to student $i$'s justified envy under the adjusted priority profile $\succ
^{A}$ results from a reversal of the original priority relation. More
precisely, suppose that

\begin{equation*}
s=\mu \left( j\right) P_{i}\mu \left( i\right) \text{, }i\succ _{s}^{A}j%
\text{ and }j\succ _{s}^{O}i\text{.}
\end{equation*}%
Then, $i$'s justified envy toward $j$ at $\mu $ under $\succ ^{A}$ is
induced by the priority adjustment at $s$. In other words, although $j$
originally has higher priority than $i$ at $s$, the adjustment reverses
their priority ranking, so that $i$ has higher priority than $j$ under $%
\succ _{s}^{A}$.

For instance, $j$ has a higher examination score than $i$. Then, $j\succ
_{s}^{O}i$. However, under an affirmative action policy, $i\succ _{s}^{A}j$,
because $j$ is a majority student and $i$ is a minority student. If $s=\mu
\left( j\right) P_{i}\mu \left( i\right) $, then $i$ has justified envy
toward $j$ at $\mu $ under $\succ _{s}^{A}$. However, since the original
priority order ranks $j$ above $i$, this envy is justified only as a result
of the affirmative action adjustment. We refer to such justified envy as
justified envy induced by the adjustment.

A matching $\mu $ is \textbf{A-fair }under $(\succ ^{O},\succ ^{A})$\textit{%
\ }if any justified envy at $\mu $ under $\succ ^{A}$ is induced by the
adjustment. Trivially, if $\mu $ is fair\textbf{\ }under $\succ ^{A}$, then
it is also A-fair\textbf{\ }under $(\succ ^{O},\succ ^{A})$.

Our focus on A-fairness is motivated by the distinction between priority
relations supported by both priority orders and those supported by only one
of them. Suppose that student $i$ prefers the school assigned to student $j$%
. If both the original and adjusted priority orders rank $i$ above $j$, then
assigning the seat to $j$ cannot be justified by either priority order. By
contrast, if the adjusted priority order ranks $i$ above $j$, while the
original priority order ranks $j$ above $i$, assigning the seat to $j$ is
supported by the original priority order. We therefore regard the former
type of justified envy as more serious than the latter. Accordingly,
A-fairness requires the elimination of justified envy supported by both
priority orders, while allowing justified envy under the adjusted priority
order when the opposite priority relation is supported by the original
priority order.

We next introduce the following counterpart to A-fairness.

\begin{definition}
Suppose that a student $i$ has justified envy toward $j$ at $\mu $ under $%
\succ ^{O}$. We say that the justified envy is\textbf{\ eliminated by the
adjustment} if $j\succ _{\mu \left( j\right) }^{A}i$.
\end{definition}

This refers to a situation in which student $i$'s justified envy toward $j$
under the original priority profile $\succ ^{O}$ is no longer justified
under the adjusted priority profile $\succ ^{A}$. More precisely, suppose
that

\begin{equation*}
s=\mu \left( j\right) P_{i}\mu \left( i\right) \text{, }i\succ _{s}^{O}j%
\text{ and }j\succ _{s}^{A}i\text{.}
\end{equation*}%
Then, $i$'s justified envy toward $j$ under $\succ ^{O}$ is eliminated by
the priority adjustment at $s$.

A matching $\mu $ is \textbf{O-fair} under $(\succ ^{O},\succ ^{A})$ if
every justified envy at $\mu $ under $\succ ^{O}$ is eliminated by the
adjustment. Clearly, if $\mu $ is fair under $\succ ^{O}$, then it is also
O-fair under $(\succ ^{O},\succ ^{A})$.

If a matching is O-fair under $(\succ ^{O},\succ ^{A})$, there may still
exist justified envy under the original priority profile. However, the
priority relation that makes such envy justified under $\succ ^{O}$ is
reversed by the adjustment. Since the adjustment is introduced to reflect
some specific policy objective, such envy is justified according to the
original priority order but is no longer justified once that policy
objective is taken into account. Therefore, we allow only such cases of
justified envy under the original priority profile.

This idea is closely related to the fairness criterion of Ayg\"{u}n and B%
\'{o} (2021). Their criterion permits a lower-scoring student to be admitted
over a higher-scoring student when the former claims at least one privilege
not claimed by the latter. Our framework also accommodates more moderate
forms of preferential treatment, such as bonus-point rules, under which
possessing a privilege does not necessarily reverse every score-based
priority comparison. Accordingly, under our notion of O-fairness, a
violation of the original priority is permitted only when the adjustment
actually reverses the priority ranking of the two students. See Section 4.1
for a detailed example.

We now characterize the O-fair and A-fair matchings under $(\succ ^{O},\succ
^{A})$.

For each $s\in S$, define $\left( \succ _{s}^{O}\cap \succ _{s}^{A}\right) $
as the binary relation over $I$ such that 
\begin{equation*}
\left( i,j\right) \in \left( \succ _{s}^{O}\cap \succ _{s}^{A}\right)
\end{equation*}%
if and only if $i\succ _{s}^{O}j$ and $i\succ _{s}^{A}j$. We also write $%
\left( \succ ^{O}\cap \succ ^{A}\right) =\left( \succ _{s}^{O}\cap \succ
_{s}^{A}\right) _{s\in S}$.

\begin{proposition}
For any matching $\mu $, the following three statements are equivalent: (1) $%
\mu $ is fair under $\left( \succ ^{O}\cap \succ ^{A}\right) $. (2) $\mu $
is A-fair under $(\succ ^{O},\succ ^{A})$. (3) $\mu $ is O-fair under $%
(\succ ^{O},\succ ^{A})$.
\end{proposition}

\textbf{Proof.} Here we only show that (1) and (2) are equivalent, as the
equivalence between (1) and (3) can be proved in a similar way. First, let $%
\mu $ be an arbitrary matching that is not fair under $\left( \succ ^{O}\cap
\succ ^{A}\right) $. Then, there are students $i$ and $j$ such that $s=\mu
\left( j\right) P_{i}\mu \left( i\right) $ and $\left( i,j\right) \in \left(
\succ _{s}^{O}\cap \succ _{s}^{A}\right) $. Since both $\succ _{s}^{O}$ and $%
\succ _{s}^{A}$ are linear orders, both $i\succ _{s}^{O}j$ and $i\succ
_{s}^{A}j$. Hence, $\mu $ is not A-fair under $(\succ ^{O},\succ ^{A})$.

Second, let $\mu $ be an arbitrary matching that is not A-fair under $(\succ
^{O},\succ ^{A})$. Then, there are students $i$ and $j$ such that $s=\mu
\left( j\right) P_{i}\mu \left( i\right) $, $i\succ _{s}^{O}j$ and $i\succ
_{s}^{A}j$. Then, $\left( i,j\right) \in \left( \succ _{s}^{O}\cap \succ
_{s}^{A}\right) $ and therefore, $\mu $ is not fair under $\left( \succ
^{O}\cap \succ ^{A}\right) $. \textbf{Q.E.D.}\newline

In light of this result, a matching is said to be \textbf{weakly stable }%
under $(\succ ^{O},\succ ^{A})$ if it is individually rational,
non-wasteful, and A-fair (or equivalently O-fair) under $(\succ ^{O},\succ
^{A})$. Trivially, any stable matching under $\succ ^{O}$ or under $\succ
^{A}$ is weakly stable\textbf{\ }under $(\succ ^{O},\succ ^{A})$. In other
words, under a weakly stable matching, any justified envy under the adjusted
priority profile must be induced by the adjustment, whereas any justified
envy under the original priority profile must be eliminated by the
adjustment.

Next, a matching $\mu $ is an $I^{\prime }$\textbf{-optimally weakly stable
matching} under $(\succ ^{O},\succ ^{A})$ if $\mu $ is weakly stable\textbf{%
\ }under $(\succ ^{O},\succ ^{A})$ and is not Pareto dominated for $%
I^{\prime }$ by any other weakly stable matching under $(\succ ^{O},\succ
^{A})$. When $I^{\prime }=I$, such a matching $\mu $ is called a \textbf{%
student-optimally weakly stable matching} (\textbf{SOWSM}) under $(\succ
^{O},\succ ^{A})$; that is, it is weakly stable under $(\succ ^{O},\succ
^{A})$ and is not Pareto dominated by any weakly\ stable matching under $%
(\succ ^{O},\succ ^{A})$.

An SOWSM under $(\succ ^{O},\succ ^{A})$ is not Pareto dominated by any SOSM
under $\succ ^{O}$ or $\succ ^{A}$. However, an SOSM under $\succ ^{O}$ (or $%
\succ ^{A}$) may be Pareto dominated by some SOWSM under $(\succ ^{O},\succ
^{A})$. In other words, by weakening the notion of stability, it may be
possible to obtain a more efficient matching.

We have the following result.

\begin{proposition}
Fix any $I^{\prime }\subseteq I$. Let $\mu ^{\ast }$ be an $I^{\prime }$%
-optimally weakly stable matching under $(\succ ^{O},\succ ^{A})$. If $\mu
^{\ast }$ is Pareto dominated for $I^{\prime }$ by $\mu ^{\prime },$ then
there is a student $i$ who has justified envy toward $j$ at $\mu ^{\prime }$
under both $\succ ^{O}$ and $\succ ^{A}$.
\end{proposition}

The proof is provided in Appendix A.1.

By this result, any Pareto improvement over an $I^{\prime }$-optimal weakly
stable matching for $I^{\prime }$ must violate fairness under $(\succ
^{O}\cap \succ ^{A})$; that is, there must exist students $i$ and $j$ such
that $i$ has justified envy toward $j$ under both $\succ ^{O}$ and $\succ
^{A}$.

\subsubsection*{Example 1}

Suppose $I=\left\{ i_{1},i_{2},i_{3},i_{4}\right\} $, $S=\left\{
s_{1},s_{2},s_{3},s_{4}\right\} $ and $q_{s}=1$ for all $s\in S$. Moreover,
students' preferences and schools' original and adjusted priority orders are
given by 
\begin{eqnarray*}
P_{i_{1}} &:&s_{2},s_{1},s_{3},s_{4}, \\
P_{i_{2}} &:&s_{1},s_{3},s_{2},s_{4}, \\
P_{i_{3}} &:&s_{1},s_{2},s_{3},s_{4}, \\
P_{i_{4}} &:&s_{1},s_{2},s_{3},s_{4},
\end{eqnarray*}%
\begin{gather*}
\succ _{s_{1}}^{O}=\succ _{s_{1}}^{A}:i_{1},i_{2},i_{3},i_{4},\text{ } \\
\succ _{s_{2}}^{O}:i_{1},i_{3},i_{2},i_{4},\text{ }\succ
_{s_{2}}^{A}:i_{2},i_{3},i_{1},i_{4}, \\
\succ _{s_{3}}^{O}:i_{1},i_{4},i_{3},i_{2},\text{ }\succ
_{s_{3}}^{A}:i_{1},i_{4},i_{2},i_{3},\text{ and} \\
\succ _{s_{4}}^{O}=\succ _{s_{4}}^{A}:i_{1},i_{2},i_{3},i_{4}.
\end{gather*}%
Here, for example, they mean $%
s_{2}P_{i_{1}}s_{1}P_{i_{1}}s_{3}P_{i_{1}}s_{4} $ and $i_{1}\succ
_{s_{1}}^{O}i_{2}\succ _{s_{1}}^{O}i_{3}\succ _{s_{1}}^{O}i_{4},$ and so on.

Then, 
\begin{gather*}
(\succ _{s_{1}}^{O}\cap \succ _{s_{1}}^{A})=\succ _{s_{1}}^{O}=\succ
_{s_{1}}^{A}, \\
(\succ _{s_{2}}^{O}\cap \succ _{s_{2}}^{A})=\left\{ \left. \left(
i_{x},i_{4}\right) \right\vert \text{ }x=1,2,3\right\} , \\
(\succ _{s_{3}}^{O}\cap \succ _{s_{3}}^{A})=\left\{ \left. \left(
i_{1},i_{x}\right) \right\vert \text{ }x=2,3,4\right\} \cup \left\{ \left.
\left( i_{4},i_{y}\right) \right\vert \text{ }y=2,3\right\} , \\
(\succ _{s_{4}}^{O}\cap \succ _{s_{4}}^{A})=\succ _{s_{4}}^{O}=\succ
_{s_{4}}^{A}.
\end{gather*}%
Thus, there are three weakly stable matchings $\mu _{1},\mu _{2}$ and $\mu
_{3},$ under $(\succ ^{O},\succ ^{A})$ such that%
\begin{eqnarray*}
\mu _{1}\left( i_{1}\right) &=&s_{1},\mu _{1}\left( i_{2}\right) =s_{2},\mu
_{1}\left( i_{3}\right) =s_{4},\mu _{1}\left( i_{4}\right) =s_{3}, \\
\mu _{2}\left( i_{1}\right) &=&s_{2},\mu _{2}\left( i_{2}\right) =s_{1},\mu
_{2}\left( i_{3}\right) =s_{4},\mu _{2}\left( i_{4}\right) =s_{3}, \\
\mu _{3}\left( i_{1}\right) &=&s_{1},\mu _{3}\left( i_{2}\right) =s_{4},\mu
_{3}\left( i_{3}\right) =s_{2},\mu _{3}\left( i_{4}\right) =s_{3}.
\end{eqnarray*}

Among them, we note that $\mu _{2}$ is an $I^{\prime }$-optimal weakly
stable matching under $(\succ ^{O},\succ ^{A})$ whenever $i_{1}\in I^{\prime
}$ or $i_{2}\in I^{\prime }$. Moreover, $\mu _{3}$ is an $\left\{
i_{3},i_{4}\right\} $-optimal weakly stable matching under $(\succ
^{O},\succ ^{A})$.

To illustrate Proposition 2, consider $I^{\prime }=\left\{
i_{3},i_{4}\right\} $. Any matching that Pareto dominates $\mu _{3}$ for $%
\left\{ i_{3},i_{4}\right\} $ must assign $s_{1}$ to either $i_{3}$ or $%
i_{4} $. In either case, one of $i_{1}$ and $i_{2}$ has justified envy
toward the student assigned to $s_{1}$ at that matching under both $\succ
^{O}$ and $\succ ^{A}$. Therefore, as stated in Proposition 2, if $\mu _{3}$
is Pareto dominated for $\left\{ i_{3},i_{4}\right\} $ by a matching, then
one of $i_{1}$ and $i_{2}$ has justified envy toward one of $i_{3}$ or $%
i_{4} $ at that matching under both $\succ ^{O}$ and $\succ ^{A}$.

For example, let $\mu _{4}$ be such that 
\begin{equation*}
\mu _{4}\left( i_{1}\right) =s_{2},\mu _{4}\left( i_{2}\right) =s_{4},\mu
_{4}\left( i_{3}\right) =s_{1},\mu _{4}\left( i_{4}\right) =s_{3},
\end{equation*}%
which Pareto dominates $\mu _{3}$. However, since $i_{2}\succ
_{s_{1}}^{O}i_{3}$ and $i_{2}\succ _{s_{1}}^{A}i_{3}$, $i_{2}$ has justified
envy toward $i_{3}$ at $\mu _{4}$ under both $\succ ^{O}$ and $\succ ^{A}$.

\subsection{EADA algorithm}

In this section, we introduce a method to derive an SOWSM\textbf{\ }under $%
(\succ ^{O},\succ ^{A})$. By Proposition 1, it suffices to obtain an SOSM
under $\left( \succ ^{O}\cap \succ ^{A}\right) $.

First, we introduce the following result.

\begin{lemma}
For any two linear priority orders $\succ _{s}^{O}$ and $\succ _{s}^{A},$
their intersection $\left( \succ _{s}^{O}\cap \succ _{s}^{A}\right) $ is a
(strict) partial order (transitive and asymmetric).\footnote{%
Note that this is satisfied even if we consider three or more partial
orders; that is, if $\succ _{s}^{1},$ $\succ _{s}^{2},$ $\cdots ,$ and $%
\succ _{s}^{n}$ are partial orders, then $\left( \succ _{s}^{1}\cap \succ
_{s}^{2}\cap \cdots \cap \succ _{s}^{n}\right) $ is also a partial order.
For the formal proof, see the previous version of this manuscript (Kitahara
and Okumura (2023)).}
\end{lemma}

\textbf{Proof. }The asymmetry of $\left( \succ _{s}^{O}\cap \succ
_{s}^{A}\right) $ is immediate. We show that $\left( \succ _{s}^{O}\cap
\succ _{s}^{A}\right) $ is transitive. Suppose $\left( i,j\right) \in \left(
\succ _{s}^{O}\cap \succ _{s}^{A}\right) $ and $\left( j,k\right) \in \left(
\succ _{s}^{O}\cap \succ _{s}^{A}\right) $. Then, $i\succ _{s}^{O}j$, $%
j\succ _{s}^{O}k$, $i\succ _{s}^{A}j$ and $j\succ _{s}^{A}k$. Since $\succ
_{s}^{O}$ and $\succ _{s}^{A}$ are transitive, $i\succ _{s}^{O}k$ and $%
i\succ _{s}^{A}k$. Therefore, $\left( i,k\right) \in \left( \succ
_{s}^{O}\cap \succ _{s}^{A}\right) $. \textbf{Q.E.D.}\newline

Although $\left( \succ _{s}^{O}\cap \succ _{s}^{A}\right) $ is transitive
and asymmetric, it may fail to be complete and hence may not be a linear
order. Therefore, we consider a school choice problem where every school $s$
has only one partial priority order denoted by $\succ _{s}$. We introduce
the (simplified) EADA algorithm, which is originally proposed by Kesten
(2010) and subsequently modified by Tang and Yu (2014), as a method to
derive an SOSM for $\succ $. The following description is almost the same as
that introduced by Tang and Yu (2014, Subsection 4.2), but they only
consider weak priority orders. Kitahara and Okumura (2026b) show that this
algorithm attains an SOSM under $\succ $ even when $\succ $ is not a
weak-order profile but a partial-order profile.

In the substeps of the algorithm, we use the student-proposing deferred
acceptance (hereafter SPDA) algorithm of Gale and Shapley (1962). However,
to apply SPDA, the schools' priority orders must be linear. Therefore, we
introduce the notion of a linear order extension of a partial priority order.

A linear order $\succ _{s}^{\ast }$ is said to be a \textbf{linear order
extension} of $\succ _{s}$ if $i\succ _{s}j$ implies $i\succ _{s}^{\ast }j$.
If $\succ _{s}$ is a partial order, then there exists at least one linear
order extension of $\succ _{s}$. See, for example, Kitahara and Okumura
(2021) on this fact. Moreover, let $\succ ^{\ast }=\left( \succ _{s}^{\ast
}\right) _{s\in S}$ be a \textbf{linear order extension profile} of a
priority profile $\succ $ if $\succ _{s}^{\ast }$ is a linear order
extension of $\succ _{s}$ for all $s\in S$.

We provide an informal description of the EADA algorithm; a formal
description is given in Appendix A.2. Given a priority profile $\succ $,
first choose a linear-order extension $\succ ^{\ast }$ and run the SPDA
algorithm under $\succ ^{\ast }$. Fix the assignments of students assigned
to underdemanded schools, as well as those who remain unmatched, and remove
these students and schools. To preserve the priority relations specified by $%
\succ $, whenever a removed student $i$ prefers a remaining school $s$ to
her fixed assignment, remove $s$ from the preference list of every remaining
student $j$ such that $i\succ _{s}j$. Run SPDA again on the resulting
subproblem and repeat this procedure until no student remains.

Since the resulting matching depends on both the priority profile $\succ $
and the selected linear-order extension profile $\succ ^{\ast }$, we denote
it by $EA\left( \succ ,\succ ^{\ast }\right) $.

The following result is due to Kitahara and Okumura (2026b).

\begin{remark}
(Kitahara and Okumura 2026, Theorem 1) If $\succ $ is a partial order
profile and $\succ ^{\ast }$ is a linear order extension profile of $\succ $%
, then $EA\left( \succ ,\succ ^{\ast }\right) $ is an SOSM for $\succ $.
\end{remark}

We immediately have the following result.

\begin{corollary}
If $\succ ^{\ast }$ is a linear order extension profile of $\left( \succ
^{O}\cap \succ ^{A}\right) $, then $EA\left( \left( \succ ^{O}\cap \succ
^{A}\right) ,\succ ^{\ast }\right) $ is an SOWSM under $\left( \succ
^{O},\succ ^{A}\right) $.
\end{corollary}

Note that both $\succ ^{O}$ and $\succ ^{A}$ are linear order extension
profiles of $\left( \succ ^{O}\cap \succ ^{A}\right) $.

We further have the following result.

\begin{proposition}
Let $\succ ^{\ast }$ be a linear order extension profile of $\left( \succ
^{O}\cap \succ ^{A}\right) ,$ $\mu ^{\ast }=EA\left( \left( \succ ^{O}\cap
\succ ^{A}\right) ,\succ ^{\ast }\right) $ and $\mu ^{\ast \ast }$ be the
SOSM under $\succ ^{\ast }$. $\mu ^{\ast }$ weakly Pareto dominates $\mu
^{\ast \ast }$. Moreover, if a student $i$ has justified envy toward $j$ at $%
\mu ^{\ast }$ under $\succ ^{\ast }$, then either $j\succ _{s}^{O}i$ or $%
j\succ _{s}^{A}i$ , where $s=\mu ^{\ast }\left( j\right) $, and $\mu ^{\ast
} $ Pareto dominates $\mu ^{\ast \ast }$ for\textbf{\ }$I$.
\end{proposition}

We introduce the proof of Proposition 3 in Appendix A.3.

To illustrate Proposition 3, we revisit Example 1. In this example, 
\begin{equation*}
EA\left( \left( \succ ^{O}\cap \succ ^{A}\right) ,\succ ^{A}\right) =\mu
_{2},
\end{equation*}%
whereas the SOSM under $\succ ^{A}$ is $\mu _{1}$. Then, $\mu _{2}$ Pareto
dominates $\mu _{1}$ for\textbf{\ }$I$. Moreover, $i_{3}$ has justified envy
toward $i_{1}$ at $\mu _{2}$ under $\succ ^{A}$, because $%
s_{2}P_{i_{3}}s_{4}=\mu _{2}\left( i_{3}\right) $ and $i_{3}\succ
_{s_{2}}^{A}i_{1}\in \mu _{2}\left( s_{2}\right) $. Moreover, $i_{1}\succ
_{s_{2}}^{O}i_{3}$, as stated in Proposition 3.

However, such justified envy is justified only according to the adjusted
priority order: if $i$ has justified envy toward $j$ under $\succ ^{A}$,
then $j\succ _{\mu ^{\ast }\left( j\right) }^{O}i$. Thus, the mechanism
achieves a Pareto improvement at the cost of introducing only those
instances of justified envy under $\succ ^{\ast }$ that are supported by the
priority relation under either $\succ ^{O}$ or $\succ ^{A}$.

We can derive an SOWSM under $\left( \succ ^{O},\succ ^{A}\right) $ by using
an algorithm other than the EADA algorithm. Kitahara and Okumura (2021) show
that if $\succ _{s}$ is a partial order for all $s\in S$, a modified version
of the stable improvement cycle algorithm \`{a} la Erdil and Ergin (2008)
also attains an SOSM under $\succ $. By Proposition 1, an SOWSM under $%
\left( \succ ^{O},\succ ^{A}\right) $ can also be obtained by utilizing the
algorithm instead of the EADA algorithm. Moreover, if the algorithm is
started from the SOSM under a linear order profile $\succ ^{\ast }$ of $%
\left( \succ ^{O}\cap \succ ^{A}\right) $, then its outcome satisfies the
same property as that stated in Proposition 3.

However, we focus on the EADA algorithm because, as shown in the next
section, the mechanism based on the EADA algorithm satisfies additional
desirable properties.

\section{Improvement and Responsive Mechanism}

Given $\left( \succ ^{O},\succ ^{A}\right) $, student $i\in I$ is said to be
an \textbf{improved student} if there exist $j\in I$ and $s\in S$ such that $%
j\succ _{s}^{O}i$ and $i\succ _{s}^{A}j$. Moreover, for $\left( \succ
^{O},\succ ^{A}\right) ,$ let $I\left( \succ ^{O},\succ ^{A}\right)
\subseteq I$ be the set of all improved students. Since $\succ _{s}^{O}\neq
\succ _{s}^{A}$ for some $s$ and they are linear orders, $I\left( \succ
^{O},\succ ^{A}\right) \neq \emptyset $.

We have the following result.

\begin{theorem}
For any $I^{\prime }\supseteq I\left( \succ ^{O},\succ ^{A}\right) $, $%
EA\left( \left( \succ ^{O}\cap \succ ^{A}\right) ,\succ ^{A}\right) $ is an $%
I^{\prime }$\textbf{-}optimally weakly stable matching under $\left( \succ
^{O},\succ ^{A}\right) $.
\end{theorem}

The proof is provided in Appendix A.4.

The theorem requires $I^{\prime }$ to include all improved students. It
shows, however, that once this condition is satisfied, no weakly stable
matching can Pareto improve upon $EA\left( \left( \succ ^{O}\cap \succ
^{A}\right) ,\succ ^{A}\right) $ for the group $I^{\prime }$, no matter
which additional students are included.

To illustrate Theorem 1, we revisit Example 1. As shown above,%
\begin{equation*}
EA\left( \left( \succ ^{O}\cap \succ ^{A}\right) ,\succ ^{A}\right) =\mu
_{2}.
\end{equation*}%
Moreover, $I\left( \succ ^{O},\succ ^{A}\right) =\left\{ i_{2},i_{3}\right\}
,$ because $i_{1}\succ _{s_{2}}^{O}i_{2}$, $i_{2}\succ _{s_{2}}^{A}i_{1}$, $%
i_{1}\succ _{s_{2}}^{O}i_{3}$, $i_{3}\succ _{s_{2}}^{A}i_{1},$ $i_{4}\succ
_{s_{3}}^{O}i_{2},$ and $i_{2}\succ _{s_{3}}^{A}i_{4}$. As noted above, $\mu
_{2}$ is an $I^{\prime }$-optimal weakly stable matching whenever $i_{2}\in
I^{\prime }$. Therefore, $EA\left( \left( \succ ^{O}\cap \succ ^{A}\right)
,\succ ^{A}\right) $ is an $I^{\prime }$-optimal weakly stable matching
under $(\succ ^{O},\succ ^{A})$ for any $I^{\prime }\supseteq I\left( \succ
^{O},\succ ^{A}\right) $, as stated in Theorem 1.

Theorem 1 is stronger than Corollary 1. Since the set of all students $I$
contains $I\left( \succ ^{O},\succ ^{A}\right) $, Theorem 1 implies, as a
special case, that the EADA outcome based on the extension profile $\succ
^{A}$ is an SOWSM. More importantly, the theorem establishes Pareto
undominance with respect to every group $I^{\prime }$ containing all
improved students. Thus, no weakly stable matching can Pareto improve upon
the EADA outcome for the students in $I^{\prime }$, even at the expense of
making students outside $I^{\prime }$ worse off.

Now, we explicitly define mechanisms. Fixing the student preference profile $%
P$ and the capacity vector $q$, a mechanism $\phi :\mathcal{L}^{\left\vert
S\right\vert }\times \mathcal{L}^{\left\vert S\right\vert }\rightarrow 
\mathcal{M}$ is a mapping where $\mathcal{M}$ is the set of all matchings. A
mechanism $\phi $ is said to be \textbf{stable }if for any $\succ ^{O},\succ
^{A}\in \mathcal{L}^{\left\vert S\right\vert },$ $\phi \left( \succ
^{O},\succ ^{A}\right) $ is stable under $\succ ^{A}$. Note that this
stability notion requires stability only with respect to the adjusted
priority profile.

By Proposition 2 and Theorem 1, if $EA\left( \left( \succ ^{O}\cap \succ
^{A}\right) ,\succ ^{A}\right) $ is Pareto dominated for $I^{\prime
}\supseteq I\left( \succ ^{O},\succ ^{A}\right) $ by $\mu ^{\prime },$ then
there is a student $i$ who has justified envy toward $j$ at $\mu ^{\prime }$
under both $\succ ^{O}$ and $\succ ^{A}$. Thus, any further Pareto
improvement for the improved students requires an additional adjustment of
priorities that resolves such justified envy.

We consider notions of responsiveness for improvements. Let $\succ ^{A\prime
}=\left( \succ _{s}^{A\prime }\right) _{s\in S}$ denote another adjusted
priority profile that differs from $\succ ^{A}$. We allow $\succ ^{A\prime
}=\succ ^{O}$. We first introduce the following notion for comparing the
strength of two priority improvements.

\begin{definition}
We say that $\succ ^{A}$ represents a \textbf{stronger improvement} \textbf{%
of} $\succ ^{O}$ \textbf{than} $\succ ^{A\prime }$ if for all $s\in S$ and
all $i,j\in I$, $i\succ _{s}^{O}j$ and $j\succ _{s}^{A\prime }i$ imply $%
j\succ _{s}^{A}i$.
\end{definition}

First, $i\succ _{s}^{O}j$ and $j\succ _{s}^{A\prime }i$ mean that the
priority of $j$ over $i$ is newly introduced under $\succ ^{A\prime }$
relative to $\succ ^{O}$. The condition requires that any such newly
improved priority relation under $\succ ^{A\prime }$ must also be present
under $\succ ^{A}$; that is, $j\succ _{s}^{A}i$. Note that if $\succ
^{A\prime }=\succ ^{O}$, then any $\succ ^{A}$ represents a stronger
improvement of $\succ ^{O}$ than $\succ ^{A\prime }$.

For example, consider an affirmative action policy that favors
minority-group students by adding bonus points to their scores. Let $\succ
^{O}$ be the priority profile based solely on examination scores, and let $%
\succ ^{A}$ and $\succ ^{A\prime }$ be priority profiles based on
examination scores together with such bonus points. Suppose that, at each
school, the bonus points assigned to minority-group students under $\succ
^{A}$ are higher than those under $\succ ^{A\prime }$. Then $\succ ^{A}$
represents a stronger improvement of $\succ ^{O}$ than $\succ ^{A\prime }$.

We obtain the following result.

\begin{lemma}
Suppose that $\succ ^{A}$ represents a stronger improvement of $\succ ^{O}$
than $\succ ^{A\prime }$. Then, 
\begin{equation}
\left( \succ _{s}^{O}\cap \succ _{s}^{A}\right) \subseteq \left( \succ
_{s}^{O}\cap \succ _{s}^{A\prime }\right)  \label{a}
\end{equation}%
for all $s\in S$.
\end{lemma}

\textbf{Proof. }Let $\left( i,j\right) \in \left( \succ _{s}^{O}\cap \succ
_{s}^{A}\right) $. Then, $i\succ _{s}^{O}j$ and $i\succ _{s}^{A}j$. Then, $%
\lnot \left( j\succ _{s}^{A\prime }i\right) $. Since $\succ _{s}^{A\prime }$
is a linear order, $i\succ _{s}^{A\prime }j$. \textbf{Q.E.D.}\newline

Intuitively, if $\succ ^{A}$ represents a stronger improvement of $\succ
^{O} $ than $\succ ^{A\prime }$, then the adjustment from $\succ ^{O}$ to $%
\succ ^{A}$ reverses at least all priority relations that are reversed under 
$\succ ^{A\prime }$. Hence, any priority relation that is preserved both
under the original priority order and under the more strongly adjusted
priority order $\succ ^{A}$ must also be preserved under $\succ ^{A\prime }$%
. In other words, as the adjustment becomes stronger, fewer of the original
priority relations remain unchanged. Therefore, for every school $s$, (\ref%
{a}) is satisfied.

Now, we introduce two notions on mechanisms of responsiveness for
improvements.

We first introduce a weaker notion. A mechanism $\phi $ is said to be 
\textbf{responsive to improvements} if, whenever $\succ ^{A}$ represents a
stronger improvement of $\succ ^{O}$ than $\succ ^{A\prime }$, the matching $%
\phi \left( \succ ^{O},\succ ^{A}\right) $ is not Pareto dominated by $\phi
\left( \succ ^{O},\succ ^{A\prime }\right) $ with respect to the set of
improved students $I\left( \succ ^{O},\succ ^{A}\right) $. This requires
that, when $\succ ^{A}$ represents a further improvement relative to $\succ
^{A\prime }$, the outcome under $\left( \succ ^{O},\succ ^{A}\right) $ must
not be Pareto dominated by the outcome under $\left( \succ ^{O},\succ
^{A\prime }\right) $ for the students whose priorities are improved under $%
\left( \succ ^{O},\succ ^{A}\right) $.

The following definition is a stronger one. A mechanism $\phi $ is said to
be \textbf{group-robust} \textbf{responsive to improvements} if, whenever $%
\succ ^{A}$ represents a stronger improvement of $\succ ^{O}$ than $\succ
^{A\prime }$, the matching $\phi \left( \succ ^{O},\succ ^{A}\right) $ is
not Pareto dominated by $\phi \left( \succ ^{O},\succ ^{A\prime }\right) $
for any group $I^{\prime }$ such that 
\begin{equation*}
I^{\prime }\supseteq I\left( \succ ^{O},\succ ^{A}\right) .
\end{equation*}

Clearly, group-robust responsiveness to improvements implies responsiveness
to improvements by taking $I^{\prime }=I\left( \succ ^{O},\succ ^{A}\right) $%
.

The significance of this stronger requirement is as follows. Consider again
an affirmative action policy that favors minority-group students by adding
bonus points to their scores. Some minority-group students may have
examination scores that are so far below those of majority students that,
even after receiving the bonus points under $\succ ^{A}$, they do not
overtake any majority student in priority at any school. Such students are
therefore not included in $I\left( \succ ^{O},\succ ^{A}\right) $.
Nevertheless, under our definition of group-robust responsiveness to
improvements, the group $I^{\prime }$ need only contain all students whose
priorities are improved, that is, $I\left( \succ ^{O},\succ ^{A}\right) $.
Since the minority group contains $I\left( \succ ^{O},\succ ^{A}\right) $,
the requirement also applies to the minority group as a whole, including
minority students whose priorities are not improved. In this sense, the
mechanism can be regarded as responsive to the entire minority group. We
explicitly consider such cases in the next section.

The following impossibility result follows from Kojima (2012, Theorem 2) by
considering priority-based affirmative action as a special case of our
framework.

\begin{remark}
There is no mechanism that is both stable and responsive to improvements.
\end{remark}

To clarify this result, we reformulate Kojima's example within our
framework. There are two minority students. Let $\succ ^{O}$ be the original
priority profile and $\succ ^{A}$ be an adjusted priority profile such that
the set of students whose priorities are improved, $I\left( \succ ^{O},\succ
^{A}\right) $ consists precisely of the two minority students. Let $\succ
^{O}=\succ ^{A\prime }$. In Kojima's example, one minority student strictly
prefers her assignment under $\phi \left( \succ ^{O},\succ ^{A\prime
}\right) $ to that under $\phi \left( \succ ^{O},\succ ^{A}\right) $, while
the other minority student receives the same assignment under both profiles.
Thus, $\phi \left( \succ ^{O},\succ ^{A\prime }\right) $ Pareto dominates $%
\phi \left( \succ ^{O},\succ ^{A}\right) $ for $I\left( \succ ^{O},\succ
^{A}\right) $.

A mechanism $\phi $ is said to be a \textbf{student-optimally} \textbf{%
weakly stable} if for any two linear order profiles $\succ ^{O}$ and $\succ
^{A},$ $\phi \left( \succ ^{O},\succ ^{A}\right) $ is an SOWSM under $\left(
\succ ^{O},\succ ^{A}\right) $. Moreover, a mechanism $\phi $ is said to be
an \textbf{improved-group}-\textbf{optimally} \textbf{weakly} \textbf{stable}
if for any two linear order profiles $\succ ^{O}$ and $\succ ^{A}$, $\phi
\left( \succ ^{O},\succ ^{A}\right) $ is an $I\left( \succ ^{O},\succ
^{A}\right) $-optimally weakly stable matching under $\left( \succ
^{O},\succ ^{A}\right) $.

Now, we introduce a mechanism%
\begin{equation*}
\phi ^{\ast }\left( \succ ^{O},\succ ^{A}\right) =EA\left( \left( \succ
^{O}\cap \succ ^{A}\right) ,\succ ^{A}\right) \text{.}
\end{equation*}

We have the following result as a corollary of Theorem 1.

\begin{corollary}
The mechanism $\phi ^{\ast }$ is improved-group-optimally weakly stable and
student-optimally weakly stable.
\end{corollary}

This result is straightforward from Theorem 1, by taking $I^{\prime
}=I\left( \succ ^{O},\succ ^{A}\right) $ and $I^{\prime }=I$, respectively.

Finally, we show that $\phi ^{\ast }$ is group-robust responsive to
improvements. The intuition behind the following result is simple. If $\succ
^{A}$ represents a stronger improvement of $\succ ^{O}$ than $\succ
^{A\prime }$, then (\ref{a}) holds for all $s\in S$. Thus, any matching that
is weakly stable under $\left( \succ ^{O},\succ ^{A\prime }\right) $ is also
weakly stable under $\left( \succ ^{O},\succ ^{A}\right) $. Theorem 1 then
implies that $\phi ^{\ast }\left( \succ ^{O},\succ ^{A}\right) $ cannot be
Pareto dominated by $\phi ^{\ast }\left( \succ ^{O},\succ ^{A\prime }\right) 
$ with respect to any group containing all students whose priorities are
improved under $\left( \succ ^{O},\succ ^{A}\right) $.

\begin{theorem}
The mechanism $\phi ^{\ast }$ is group-robust responsive to improvements.
\end{theorem}

We introduce the formal proof of Theorem 2 in Appendix A.5.

By Remark 2, no stable mechanism can satisfy even responsiveness to
improvements. In contrast, once stability is relaxed to weak stability, the
mechanism $\phi ^{\ast }$ satisfies the stronger requirement of group-robust
responsiveness to improvements.

To illustrate Theorem 2, we reconsider Example 1. We can show this by
reconsidering Example 1. We assume 
\begin{equation}
\succ _{s_{2}}^{A\prime }:i_{3},i_{2},i_{1},i_{4}\text{ and }\succ
_{s_{x}}^{A\prime }=\succ _{s_{x}}^{A}\text{ for all }x=1,3,4.  \label{e}
\end{equation}%
Then, $\succ ^{A}$ represents a stronger improvement of $\succ ^{O}$ than $%
\succ ^{A\prime }$, because $i\succ _{s}^{O}j$ and $j\succ _{s}^{A\prime }i$
imply $j\succ _{s}^{A}i$. Then, $\succ ^{A}$ represents a stronger
improvement of $\succ ^{O}$ than $\succ ^{A\prime },$ because $i_{1}\succ
_{s_{2}}^{O}i_{2}$, $i_{2}\succ _{s_{2}}^{A\prime }i_{1},$ $i_{2}\succ
_{s_{2}}^{A}i_{1},$ $i_{1}\succ _{s_{2}}^{O}i_{3}$, $i_{3}\succ
_{s_{2}}^{A\prime }i_{1}$, and $i_{3}\succ _{s_{2}}^{A}i_{1}$. Note that $%
i_{3}\succ _{s_{2}}^{O}i_{2}$ and $i_{3}\succ _{s_{2}}^{A\prime }i_{2},$ but 
$i_{2}\succ _{s_{2}}^{A}i_{3}$. Then, $\succ ^{A\prime }$ does not represent
a stronger improvement of $\succ ^{O}$ than $\succ ^{A}$.

In this case, 
\begin{equation*}
\phi ^{\ast }\left( \succ ^{O},\succ ^{A\prime }\right) =EA\left( \left(
\succ ^{O}\cap \succ ^{A\prime }\right) ,\succ ^{A\prime }\right) =\mu _{3}.
\end{equation*}%
Since $I\left( \succ ^{O},\succ ^{A}\right) =\left\{ i_{2},i_{3}\right\} $
and 
\begin{equation*}
\mu _{2}\left( i_{2}\right) =s_{1}P_{i_{2}}s_{4}=\mu _{3}\left( i_{2}\right)
,
\end{equation*}%
$\phi ^{\ast }\left( \succ ^{O},\succ ^{A}\right) $ is not Pareto dominated
by $\phi ^{\ast }\left( \succ ^{O},\succ ^{A\prime }\right) $ for any $%
I^{\prime }\supseteq I\left( \succ ^{O},\succ ^{A}\right) $, as stated in
Theorem 2.

\section{Applications}

\subsection{Priority-based Affirmative Action}

We assume that students are divided into minority group $I^{m}$ and majority
group $I^{M}$, where $I^{m}\cap I^{M}=\emptyset $ and $I^{m}\cup I^{M}=I$.
The score of student $i$ at school $s$ is denoted by $\sigma _{i}^{s}\in 
\mathbb{R}_{+}$. We arbitrarily fix $\sigma _{i}^{s}\in \mathbb{R}_{+}$ for
all $s\in S$ and all $i\in I$. In this section, the original priority
profile is determined by $\left( \sigma _{i}^{s}\right) _{i\in I,s\in S}$.
That is, for all $s\in S$, let $\succ _{s}^{O}$ be such that $i\succ
_{s}^{O}j$ if and only if $\sigma _{i}^{s}>\sigma _{j}^{s}$.

Let $b_{i}^{s}\in \mathbb{R}_{+}$ be the bonus point for student $i$ at
school $s$. Moreover, let $b=\left( b_{i}^{s}\right) _{s\in S,i\in I}\in 
\mathbb{R}_{+}^{\left\vert I\right\vert \times \left\vert S\right\vert }$.
The adjusted priority profile is determined by $\left( \sigma
_{i}^{s}+b_{i}^{s}\right) _{i\in I,s\in S}$; that is, for all $s\in S$, let $%
\succ _{s}^{A}$ be such that $i\succ _{s}^{A}j$ if and only if $\sigma
_{i}^{s}+b_{i}^{s}>\sigma _{j}^{s}+b_{j}^{s}$. To avoid a trivial case, we
assume that $\succ ^{O}$ and $\succ ^{A}$ are distinct.

Here, we assume that $\sigma _{i}^{s}\neq \sigma _{j}^{s}$ and $\sigma
_{i}^{s}+b_{i}^{s}\neq \sigma _{j}^{s}+b_{j}^{s}$ for all distinct $i,j\in I$
and all $s\in S$. These assumptions ensure that both $\succ _{s}^{O}$ and $%
\succ _{s}^{A}$ are linear orders.

We have the following result.

\begin{proposition}
Suppose that $b_{i}^{s}=0$ for all $i\in I^{M}$ and all $s\in S$. Then, $%
\phi ^{\ast }\left( \succ ^{O},\succ ^{A}\right) $ is an $I^{m}$-optimally
weakly stable matching.
\end{proposition}

A sufficient condition for $I^{m}$-optimal weak stability is that $%
b_{i}^{s}=0$ for all $i\in I^{M}$ and all $s\in S$. This implies that $%
I^{M}\cap I\left( \succ ^{O},\succ ^{A}\right) =\emptyset $. Note that there
may exist some minority student whose priority is not improved at any
school; that is, $I\left( \succ ^{O},\succ ^{A}\right) \varsubsetneq I^{m}$
may hold. However, even in that case, since Theorem 1 requires only $I\left(
\succ ^{O},\succ ^{A}\right) \subseteq I^{m}$, the $I^{m}$\textbf{-}optimal
weak stability of $\phi ^{\ast }\left( \succ ^{O},\succ ^{A}\right) $ is
ensured.

By Propositions 2 and 4 and Theorem 1, we immediately have the following
result.

\begin{corollary}
If a matching $\mu ^{\prime }$ Pareto dominates $\phi ^{\ast }\left( \succ
^{O},\succ ^{A}\right) $ for the minority students, then there is a student $%
i$ who has justified envy toward $j$ at $\mu ^{\prime }$ under both $\succ
^{O}$ and $\succ ^{A}$.
\end{corollary}

Thus, any Pareto improvement over $\phi ^{\ast }\left( \succ ^{O},\succ
^{A}\right) $ for the minority students requires accepting an instance of
justified envy under both $\succ ^{O}$ and $\succ ^{A}$. Equivalently, it
requires a priority violation that can be justified by neither priority
order.

Next, we can consider other priority orders of $s$ based on $\left(
b_{i}^{s}\right) _{i\in I}$ and $\left( \sigma _{i}^{s}\right) _{i\in I}$.
Let $\succ _{s}^{\delta _{s}}$ be such that $i\succ _{s}^{\delta _{s}}j$ if
and only if 
\begin{equation*}
\left( 1-\delta _{s}\right) \sigma _{i}^{s}+\delta _{s}\left(
b_{i}^{s}+\sigma _{i}^{s}\right) =\sigma _{i}^{s}+\delta
_{s}b_{i}^{s}>\sigma _{j}^{s}+\delta _{s}b_{j}^{s}=\left( 1-\delta
_{s}\right) \sigma _{j}^{s}+\delta _{s}\left( b_{j}^{s}+\sigma
_{j}^{s}\right) ,
\end{equation*}%
where $\delta _{s}\in \left[ 0,1\right] $. That is, the priority order $%
\succ _{s}^{\delta _{s}}$ is based on the weighted sum of the original score
and the adjusted score. Trivially, $\succ _{s}^{0}=\succ _{s}^{O}$ and $%
\succ _{s}^{1}=\succ _{s}^{A}$.

Next, we consider another bonus point profile. Let $b_{i}^{s\prime }\in 
\mathbb{R}_{+}$ be the bonus point for student $i$ at school $s$ and $%
b^{\prime }=\left( b_{i}^{s\prime }\right) _{s\in S,i\in I}\in \mathbb{R}%
_{+}^{\left\vert I\right\vert \times \left\vert S\right\vert }$. Let $\succ
_{s}^{A\prime }$ be such that $i\succ _{s}^{A\prime }j$ if and only if $%
\sigma _{i}^{s}+b_{i}^{s\prime }>\sigma _{j}^{s}+b_{j}^{s\prime }$. To avoid
a trivial case, we assume that $\succ ^{O}$, $\succ ^{A}$ and $\succ
^{A\prime }$ are distinct.

Suppose 
\begin{eqnarray}
b_{i}^{s} &=&b_{i}^{s\prime }=0\text{ for all }i\in I^{M}\text{ and all }%
s\in S,  \label{a1} \\
b_{i}^{s} &\geq &b_{i}^{s\prime }\geq 0\text{ for all }i\in I^{m}\text{ and
all }s\in S,  \label{a2} \\
b_{i}^{s} &>&b_{i}^{s\prime }\text{ for some }i\in I^{m}\text{ and some }%
s\in S.  \label{a3}
\end{eqnarray}%
These three assumptions do not ensure that $\succ ^{A}$ represents a
stronger improvement of $\succ ^{O}$ than $\succ ^{A\prime }$. For example,
suppose, for $i,j\in I^{m},$ 
\begin{equation*}
\sigma _{i}^{s}=10,\text{ }\sigma _{j}^{s}=11,\text{ }b_{i}^{s\prime }=3,%
\text{ }b_{j}^{s\prime }=1,\text{ }b_{i}^{s}=4,\text{ and }b_{j}^{s}=5.
\end{equation*}%
Then, $j\succ _{s}^{O}i$ and $i\succ _{s}^{A\prime }j$, but $j\succ
_{s}^{A}i $. Thus, we consider an additional assumption as follows. For all $%
s\in S$ and all $i,j\in I$,%
\begin{equation}
b_{i}^{s\prime }>b_{j}^{s\prime }\Rightarrow b_{i}^{s\prime }-b_{j}^{s\prime
}\leq b_{i}^{s}-b_{j}^{s}.  \label{a4}
\end{equation}

Assumption (\ref{a4}) means that if student $i$ receives a larger bonus than
student $j$ under $b^{\prime }$, then $i$'s bonus-point advantage over $j$
is at least as large under $b$. In other words, under this policy, student $%
i $ is treated as a student who should receive stronger preferential
treatment than student $j$. Assumption (\ref{a4}) requires that this
relative policy priority be preserved when the priority adjustment is
strengthened from $b^{\prime }$ to $b$.

We have the following result.

\begin{proposition}
If $b,b^{\prime }\in \mathbb{R}_{+}^{\left\vert I\right\vert \times
\left\vert S\right\vert }$ satisfy (\ref{a1}) to (\ref{a4}), then $\succ
^{A} $ represents a stronger improvement of $\succ ^{O}$ than $\succ
^{A\prime }$.
\end{proposition}

\textbf{Proof. }If $j\succ _{s}^{O}i$ and $i\succ _{s}^{A\prime }j,$ then 
\begin{equation*}
0<\sigma _{j}^{s}-\sigma _{i}^{s}<b_{i}^{s\prime }-b_{j}^{s\prime }.
\end{equation*}%
By $b_{i}^{s\prime }>b_{j}^{s\prime }$ and Assumption (\ref{a4}), 
\begin{equation*}
\sigma _{j}^{s}-\sigma _{i}^{s}<b_{i}^{s\prime }-b_{j}^{s\prime }\leq
b_{i}^{s}-b_{j}^{s}.
\end{equation*}%
Therefore, 
\begin{equation*}
\sigma _{j}^{s}+b_{j}^{s}<\sigma _{i}^{s}+b_{i}^{s},
\end{equation*}%
which implies $i\succ _{s}^{A}j$. \textbf{Q.E.D.}\newline

Therefore, if $b,b^{\prime }\in \mathbb{R}_{+}^{\left\vert I\right\vert
\times \left\vert S\right\vert }$ satisfy Assumptions (\ref{a1})-(\ref{a4}),
then $b$ can be interpreted as a stronger priority adjustment in favor of
minority students than $b^{\prime }$. By Theorem 2 and Proposition 5, we
obtain that the matching induced by the weaker adjustment does not Pareto
dominate the matching induced by the stronger adjustment for the minority
group.

\subsection{Component Tests and Their Weighted-Average Priority}

Suppose that each school evaluates students using two component tests, $X$
and $Y$. For example, $X$ may be a nationwide standardized academic
examination, whereas $Y$ may be a school-specific examination. The score of
student $i$ at school $s$ of component $Z\in \left\{ X,Y\right\} $ is
denoted by $\sigma _{i}^{Z_{s}}\in \mathbb{R}_{+}$. Fix $\sigma
_{i}^{Z_{s}}\in \mathbb{R}_{+}$ for all $s\in S$, all $i\in I$ and all $Z\in
\left\{ X,Y\right\} $.

In some school admissions systems, students are evaluated using a weighted
average of their scores on these components.\footnote{%
Such weighted-score evaluation is commonly used in admissions to Japanese
national universities, where applicants are evaluated based on a combination
of their scores on the nationwide Common Test for University Admissions and
a university-specific examination.} Thus, for each $\delta _{s}\in \left[ 0,1%
\right] $, we define student $i$'s composite score at school $s$ by 
\begin{equation*}
\sigma _{i}^{s}\left( \delta _{s}\right) =\delta _{s}\sigma
_{i}^{X_{s}}+\left( 1-\delta _{s}\right) \sigma _{i}^{Y_{s}}.
\end{equation*}%
Define $\succ _{s}^{\delta _{s}}$ by 
\begin{equation*}
i\succ _{s}^{\delta _{s}}j\text{ }\Leftrightarrow \sigma _{i}^{s}\left(
\delta _{s}\right) >\sigma _{j}^{s}\left( \delta _{s}\right) .
\end{equation*}%
This relation $\succ _{s}^{\delta _{s}}$ need not be linear order, because $%
\sigma _{i}^{s}\left( \delta _{s}\right) =\sigma _{j}^{s}\left( \delta
_{s}\right) $ may hold for some $\delta _{s}\in \left[ 0,1\right] $.

Suppose that each school has an admissible range of weights, with a given
lower bound $\underline{\delta }_{s}$ and upper bound $\overline{\delta }%
_{s} $, where 
\begin{equation*}
0\leq \underline{\delta }_{s}\leq \overline{\delta }_{s}\leq 1.
\end{equation*}%
A weight $\delta _{s}$ is admissible for school $s$ if $\delta _{s}\in \left[
\underline{\delta }_{s},\overline{\delta }_{s}\right] $. We assume that the
priority orders induced by the two endpoint weights, $\succ _{s}^{\underline{%
\delta }_{s}}$ and $\succ _{s}^{\overline{\delta }_{s}}$, are linear. For an
interior admissible weight $\delta _{s}$, we explicitly impose linearity
whenever it is required.

To illustrate the problem, we introduce the following example. There are
four students, whose preferences are the same as in Example 1, and four
schools. Each school evaluates students using a nationwide standardized
academic examination, $X$, and a school-specific examination, $Y$. The
students' scores are summarized in the following table.

\begin{center}
$%
\begin{array}{ccccc}
& i_{1} & i_{2} & i_{3} & i_{4} \\ 
\sigma ^{X} & 60 & 100 & 80 & 40 \\ 
\sigma ^{Y_{s_{1}}} & 100 & 80 & 60 & 40 \\ 
\sigma ^{Y_{s_{2}}} & 100 & 70 & 90 & 80 \\ 
\sigma ^{Y_{s_{3}}} & 100 & 40 & 50 & 90 \\ 
\sigma ^{Y_{s_{4}}} & 100 & 80 & 60 & 40%
\end{array}%
$
\end{center}

Since $X$ is a standardized academic examination, its score does not depend
on the school; that is, $\sigma _{i}^{Xs_{j}}=\sigma _{i}^{X}$ for all $i\in
\left\{ i_{1},\cdots ,i_{4}\right\} $ and $j\in \left\{ 1,\cdots ,4\right\} $%
. The lower and upper bounds on the admissible weight assigned to $X$ by
each school are summarized in the following table.

\begin{center}
$%
\begin{array}{ccccc}
& s_{1} & s_{2} & s_{3} & s_{4} \\ 
\underline{\delta }_{s} & 0.2 & 0.3 & 0.2 & 0.2 \\ 
\overline{\delta }_{s} & 0.3 & 0.6 & 0.4 & 0.3%
\end{array}%
$
\end{center}

Thus, the priority orders induced by the extreme admissible weights coincide
with the original and adjusted priority orders introduced in Example 1. More
precisely, they are given as follows:%
\begin{gather*}
\succ _{s_{1}}^{\underline{\delta }}=\succ _{s_{1}}^{\overline{\delta }%
}:i_{1},i_{2},i_{3},i_{4},\text{ } \\
\succ _{s_{2}}^{\underline{\delta }}:i_{1},i_{3},i_{2},i_{4},\text{ }\succ
_{s_{2}}^{\overline{\delta }}:i_{2},i_{3},i_{1},i_{4}, \\
\succ _{s_{3}}^{\underline{\delta }}:i_{1},i_{4},i_{3},i_{2},\text{ }\succ
_{s_{3}}^{\overline{\delta }}:i_{1},i_{4},i_{2},i_{3},\text{ and} \\
\succ _{s_{4}}^{\underline{\delta }}=\succ _{s_{4}}^{\overline{\delta }%
}:i_{1},i_{2},i_{3},i_{4}.
\end{gather*}

Suppose that the upper-bound weight profile $\bar{\delta}$ is adopted. Then,
the SOSM under $\succ ^{\overline{\delta }}$ is $\mu _{1}$. As stated
earlier, this matching is Pareto dominated by $\mu _{2}$ for $I$. At $\mu
_{2}$, student $i_{3}$ has justified envy toward $i_{1}$ at $\mu _{2}$ under 
$\succ ^{\overline{\delta }}$, because $s_{2}P_{i_{3}}s_{4}=\mu _{2}\left(
i_{3}\right) $ and 
\begin{gather*}
\sigma _{i_{3}}^{s_{2}}\left( 0.6\right) =0.6\times 80+0.4\times 90=84 \\
>76=0.6\times 60+0.4\times 100=\sigma _{i_{1}}^{s_{2}}\left( 0.6\right) .
\end{gather*}%
However, $\sigma _{i_{3}}^{s_{2}}\left( \delta _{s_{2}}\right) <\sigma
_{i_{1}}^{s_{2}}\left( \delta _{s_{2}}\right) $ whenever $\delta
_{s_{2}}<1/3 $. Since each $\delta _{s_{2}}\in \left[ 0.3,0.6\right] $ is
admissible, there are admissible weights---for example, any $\delta
_{s_{2}}\in \left[ 0.3,1/3\right) $---under which $i_{3}$ has a lower
composite score than $i_{1}$. Thus, this justified envy at $\mu _{2}$ under $%
\succ ^{\overline{\delta }}$ is not robust to all admissible weights. This
implies that $\mu _{2}$ does not generate justified envy that persists
throughout the admissible range and is therefore acceptable under the
proposed criterion.

On the other hand, at $\mu _{4}$, 
\begin{equation*}
\mu _{4}\left( i_{3}\right) =s_{1}P_{i_{2}}s_{4}=\mu _{4}\left( i_{2}\right) 
\text{ and }\sigma _{i_{2}}^{s_{1}}\left( \delta _{s_{1}}\right) >\sigma
_{i_{3}}^{s_{1}}\left( \delta _{s_{1}}\right)
\end{equation*}
for any admissible weight $\delta _{s_{1}}\in \left[ 0.2,0.3\right] $. Thus, 
$i_{2}$ has justified envy toward $i_{3}$ at $\mu _{4}$ under every
admissible weight. Consequently, this justified envy is robust, and $\mu
_{4} $ is not permitted under the proposed rule.

As shown below, by applying our preceding results, we introduce a method for
identifying a Pareto-undominated matching among those that are acceptable in
this sense.

We have the following result.

\begin{lemma}
For any $\delta _{s}\in \left[ \underline{\delta }_{s},\overline{\delta }_{s}%
\right] $ such that $\succ _{s}^{\delta _{s}}$ is a linear order, $\succ
_{s}^{\delta _{s}}$ is a linear order extension of $\left( \succ _{s}^{%
\underline{\delta }_{s}}\cap \succ _{s}^{\overline{\delta }_{s}}\right) .$
\end{lemma}

\textbf{Proof. }Suppose $\left( i,j\right) \in \left( \succ _{s}^{\underline{%
\delta }_{s}}\cap \succ _{s}^{\overline{\delta }_{s}}\right) $; that is, $%
i\succ _{s}^{\underline{\delta }_{s}}j$ and $i\succ _{s}^{\overline{\delta }%
_{s}}j$. For any $\delta _{s}\in \left[ \underline{\delta }_{s},\overline{%
\delta }_{s}\right] ,$ there is $\lambda \in \left[ 0,1\right] $ such that 
\begin{equation*}
\delta _{s}=\lambda \underline{\delta }_{s}+\left( 1-\lambda \right) 
\overline{\delta }_{s}.
\end{equation*}%
Since $\sigma _{i}^{s}\left( \delta _{s}\right) $ is affine in $\delta _{s},$
we have 
\begin{equation*}
\sigma _{i}^{s}\left( \delta _{s}\right) -\sigma _{j}^{s}\left( \delta
_{s}\right) =\lambda \left( \sigma _{i}^{s}\left( \underline{\delta }%
_{s}\right) -\sigma _{j}^{s}\left( \underline{\delta }_{s}\right) \right)
+\left( 1-\lambda \right) \left( \sigma _{i}^{s}\left( \overline{\delta }%
_{s}\right) -\sigma _{j}^{s}\left( \overline{\delta }_{s}\right) \right) .
\end{equation*}%
Both terms in brackets are strictly positive. Hence, $\sigma _{i}^{s}\left(
\delta _{s}\right) >\sigma _{j}^{s}\left( \delta _{s}\right) $. \textbf{%
Q.E.D.}\newline

We consider the following priority profile 
\begin{equation*}
\succ ^{\delta }=\left( \succ _{s}^{\delta _{s}}\right) _{s\in S},\text{ }%
\succ ^{\underline{\delta }}=\left( \succ _{s}^{\underline{\delta }%
_{s}}\right) _{s\in S},\text{ and }\succ ^{\bar{\delta}}=\left( \succ _{s}^{%
\bar{\delta}_{s}}\right) _{s\in S}.
\end{equation*}

\begin{lemma}
Fix $\delta \in \left[ 0,1\right] ^{\left\vert S\right\vert }$ satisfying $%
\delta _{s}\in \left[ \underline{\delta }_{s},\overline{\delta }_{s}\right] $
for all $s\in S$. Suppose that $\succ ^{\delta }$ is a linear order profile.
If a matching $\mu $ is stable under $\succ ^{\delta }$, then $\mu $ is
weakly stable under $\left( \succ ^{\underline{\delta }},\succ ^{\bar{\delta}%
}\right) $.
\end{lemma}

\textbf{Proof. }Let $\mu $ be a stable matching under $\succ ^{\delta }$. By
Lemma 3, $\succ ^{\delta }$ is a linear order extension profile of $\left(
\succ ^{\underline{\delta }}\cap \succ ^{\overline{\delta }}\right) $. Thus, 
$\mu $ is also stable under $\left( \succ ^{\underline{\delta }}\cap \succ ^{%
\bar{\delta}}\right) $. By Proposition 1, $\mu $ is weakly stable under $%
\left( \succ ^{\underline{\delta }},\succ ^{\bar{\delta}}\right) $. \textbf{%
Q.E.D.}\newline

As a conclusion of this section, we provide the following result.

\begin{proposition}
Fix $\delta \in \left[ 0,1\right] ^{\left\vert S\right\vert }$ satisfying $%
\delta _{s}\in \left[ \underline{\delta }_{s},\overline{\delta }_{s}\right] $
for all $s\in S$. Suppose that $\succ ^{\delta }$ is a linear order and
define 
\begin{equation*}
\mu ^{\ast }\left( \delta \right) =EA\left( \left( \succ ^{\underline{\delta 
}}\cap \succ ^{\bar{\delta}}\right) ,\succ ^{\delta }\right) .
\end{equation*}%
Then,

\begin{description}
\item[(i)] If $\succ ^{\delta }$ is a linear order profile, then $\mu ^{\ast
}\left( \delta \right) $ is SOWSM under $\left( \succ ^{\underline{\delta }%
},\succ ^{\bar{\delta}}\right) $ and weakly Pareto dominates any stable
matching under $\succ ^{\delta }$.

\item[(ii)] If there are students $i$ and $j$ such that $\mu ^{\ast }\left(
\delta \right) \left( j\right) =sP_{i}\mu ^{\ast }\left( \delta \right)
\left( i\right) $ and $\sigma _{i}^{s}\left( \delta _{s}\right) >\sigma
_{j}^{s}\left( \delta _{s}\right) $, then $\mu ^{\ast }\left( \delta \right) 
$ Pareto dominates any stable matching under $\succ ^{\delta }$ and either $%
\sigma _{i}^{s}\left( \underline{\delta }_{s}\right) <\sigma _{j}^{s}\left( 
\underline{\delta }_{s}\right) $ or $\sigma _{i}^{s}\left( \overline{\delta }%
_{s}\right) <\sigma _{j}^{s}\left( \overline{\delta }_{s}\right) $.

\item[(iii)] If $\mu ^{\ast }\left( \delta \right) $ is Pareto dominated by $%
\mu ^{\prime },$ then there are students $i$ and $j$ such that $\mu ^{\prime
}\left( j\right) =sP_{i}\mu ^{\prime }\left( i\right) $ and $\sigma
_{i}^{s}\left( \delta _{s}\right) >\sigma _{j}^{s}\left( \delta _{s}\right) $
for all $\delta _{s}\in \left[ \underline{\delta }_{s},\overline{\delta }_{s}%
\right] .$
\end{description}
\end{proposition}

By Lemma 3, if $\succ ^{\delta }$ is a linear order profile, then $\succ
^{\delta }$ is a linear order extension profile of $\left( \succ ^{%
\underline{\delta }}\cap \succ ^{\overline{\delta }}\right) $. By
Proposition 3, $\mu ^{\ast }\left( \delta \right) $ is SOWSM under $\left(
\succ ^{\underline{\delta }},\succ ^{\bar{\delta}}\right) $. Moreover, $\mu
^{\ast }\left( \delta \right) $ weakly Pareto dominates the SOSM under $%
\succ ^{\delta }$ that weakly Pareto dominates any stable matching under $%
\succ ^{\delta }$. Therefore, we have the first result of Proposition 6.
Next, the conditions on the second result mean that $i$ has a justified envy
toward $j$. Thus, by Proposition 3 and Lemma 3, we have the second result.
Finally, since $\mu ^{\ast }\left( \delta \right) $ is SOWSM under $\left(
\succ ^{\underline{\delta }},\succ ^{\bar{\delta}}\right) ,$ by Proposition
2, there are students $i$ and $j$ such that $\mu ^{\prime }\left( j\right)
=sP_{i}\mu ^{\prime }\left( i\right) $ and both $i\succ _{s}^{\underline{%
\delta }_{s}}j$ and $i\succ _{s}^{\overline{\delta }_{s}}j$. Then, $\left(
i,j\right) \in \left( \succ _{s}^{\underline{\delta }}\cap \succ _{s}^{\bar{%
\delta}}\right) $ and Lemma 3 imply that $i\succ _{s}^{\delta _{s}}j$ for
all $\delta _{s}\in \left[ \underline{\delta }_{s},\overline{\delta }_{s}%
\right] $ and thus, $\sigma _{i}^{s}\left( \delta _{s}\right) >\sigma
_{j}^{s}\left( \delta _{s}\right) $ for all $\delta _{s}\in \left[ 
\underline{\delta }_{s},\overline{\delta }_{s}\right] $.

This proposition shows that our approach can improve student welfare when
schools' priorities are determined by weighted-average scores. For any
admissible weight profile that induces linear priority orders, the proposed
matching weakly Pareto dominates every matching that is stable under the
corresponding weighted-average priority profile. Although this welfare
improvement may create justified envy under some admissible weights, any
such envy is reversed at least one endpoint of the admissible interval and
therefore does not persist throughout the entire admissible range. By
contrast, any further Pareto improvement upon the proposed matching must
generate justified envy that persists throughout the entire admissible
range. Thus, further improving student welfare would require violating a
priority comparison that cannot be reversed by changing the weights within
their admissible ranges.

\section{Concluding Remarks}

Our analysis assumes that both the original and adjusted priority orders are
strict linear orders and thus admit no ties. This assumption is essential
for the equivalence between A-fairness and O-fairness established in
Proposition 1. Indeed, suppose that two students $i$ and $j$ are tied under
a school's original priority order, whereas $i$ is ranked strictly above $j$
under its adjusted priority order. If $j$ is assigned to the school and $i$
prefers that school to their current assignment, the resulting matching can
be O-fair but never be A-fair. Thus, the equivalence in Proposition 1 need
not extend to priority orders that allow ties, and the notion of weak
stability must be reconsidered in such environments. Moreover, even under a
reformulated notion of weak stability, the main results established in
Theorems 1 and 2 may not generally extend to the setting. Kitahara and
Okumura (2023) briefly discuss these issues.

Responsiveness to priority-based affirmative action is also studied Kojima
(2012) and Jiao and Shen (2021). A distinctive feature of our approach is
that we explicitly introduce two priority profiles. Appendix A.6 discusses
in detail how our notion of responsiveness relates to that of Jiao and Shen
(2021).

\section*{References}

\begin{description}
\item Abdulkadiroglu, A. and Grigoryan, A., 2025. Market Design for
Distributional Objectives in Allocation Problems: An Axiomatic Approach
Available at SSRN: http://dx.doi.org/10.2139/ssrn.5284396

\item Afacan, M.O., Salman, U. 2016. Affirmative actions: The Boston
mechanism case. Economics Letters 141, 95--97.

\item Ayg\"{u}n, O., B\'{o}, I., 2021. College Admission with
Multidimensional Privileges: The Brazilian Affirmative Action Case, American
Economic Journal: Microeconomics 13(3), 1--28.

\item Balinski, M., S\"{o}nmez, T. 1999. A tale of two mechanisms: student
placement,\ Journal of Economic Theory 84(1), 73--94.

\item Cerrone, C., Hermstr\"{u}wer, Y., Kesten, O. 2024. School Choice with
Consent: An Experiment, Economic Journal 134(661), 1760--1805.

\item Dur, U., Xie, Y. 2023. Responsiveness to priority-based affirmative
action policy in school choice. Journal of Public Economic Theory 25(2),
229-244

\item Erdil, A., Ergin, H. 2008. What's the matter with tie-breaking?
Improving efficiency in school choice,\ American Economic Review 98(3),
669--689.

\item Gale D., Shapley L.S. 1962. College admissions and the stability of
marriage. American Mathematical Monthly 69(1), 9--15.

\item Hirata, D., Kasuya, Y., Okumura, Y. 2022. Stability,
Strategy-Proofness, and Respect for Improvements, Mimeo Available at SSRN:
https://ssrn.com/abstract=3876865

\item Jiao, Z., Shen, Z. 2021. School choice with priority-based affirmative
action: A responsive solution, Journal of Mathematical Economics 92, 1--9.

\item Jiao, Z., Shen, Z., Tian, G. 2022. When is the deferred acceptance
mechanism responsive to priority-based affirmative action? Social Choice and
Welfare 58, 257--282.

\item Jiao, Z., Tian, G., 2018. Two further impossibility results on
responsive affirmative action in school choice. Economics Letters 166,
60--62.

\item Kesten, O. 2010. School choice with consent. Quarterly Journal of
Economics 125(3), 1297--1348.

\item Kitahara, M., Okumura, Y. 2020. Stable Improvement Cycles in a
Controlled School Choice. Mimeo Available at SSRN:
https://ssrn.com/abstract=3582421

\item Kitahara, M., Okumura, Y. 2021. Improving Efficiency in School Choice
under Partial Priorities, International Journal of Game Theory 50, 971--987.

\item Kitahara, M., Okumura, Y. 2023. School Choice with Multiple
Priorities. Mimeo Available at arXiv:2308.04780

\item Kitahara, M., Okumura, Y. 2026a. A Stable and Strategy-Proof
Controlled School Choice Mechanism with Integrated and Flexible Rules,
International Journal of Game Theory 55, 25

\item Kitahara, M., Okumura, Y. 2026b. Stable Improvement Cycle Mechanism
Versus Efficiency-adjusted Deferred Acceptance Mechanism beyond Weak
Priority Orders Mimeo

\item Kojima, F. 2012. School choice: Impossibilities for affirmative
action. Games and Economic Behavior 75(2), 685--693.

\item Kurata, R., Hamada, N., Iwasaki, A. and Yokoo, M. 2017. Controlled
school choice with soft bounds and overlapping types,\ Journal of Artificial
Intelligence Research 58, 153--184.

\item S\"{o}nmez, T., Yenmez, M.B. 2022. Affirmative action in India via
vertical, horizontal, and overlapping reservations,\ Econometrica\ 90(3),
1143-1176.

\item Tang, Q., Yu, J., 2014. A new perspective on Kesten's school choice
with consent idea. Journal of Economic Theory 154, 543--561.

\item Tang, Q., Zhang, Y. 2021. Weak stability and Pareto efficiency in
school choice. Economic Theory 71, 533--552.

\item Wang, T. 2009. Preferential policies for minority college admission in
China: Recent developments, necessity, and impact,\ In M. Zhou and A.M. Hill
(eds) Affirmative action in China and the U.S., Palgrave Macmillan.\newpage
\end{description}

\section*{Appendix}

\subsection*{A.1 Proof of Proposition 2}

Suppose, to the contrary, that there exists a matching $\mu ^{\prime }$ that
Pareto dominates $\mu ^{\ast }$ for $I^{\prime }$, but no student has
justified envy toward another student at $\mu ^{\prime }$ under both $\succ
^{O}$ and $\succ ^{A}$. By Proposition 1, $\mu ^{\prime }$ is fair under $%
\left( \succ ^{O}\cap \succ ^{A}\right) $.

We first modify $\mu ^{\prime }$ to obtain an individually rational
matching. Define $\mu ^{\prime \prime }$ by 
\[
\mu''(i)=
\begin{cases}
\emptyset & \text{if } \emptyset P_i \mu'(i),\\
\mu'(i)   & \text{if } \mu'(i) R_i \emptyset.
\end{cases}
\]

Then, $\mu ^{\prime \prime }$ is individually rational and weakly Pareto
dominates $\mu ^{\prime }$ for $I$. Therefore, $\mu ^{\prime \prime }$ also
Pareto dominates $\mu ^{\ast }$ for $I^{\prime }$.

We next show that $\mu ^{\prime \prime }$ is fair under $\left( \succ
^{O}\cap \succ ^{A}\right) $. Suppose, to the contrary, that some student $i$
has justified envy toward another student $j$ at $\mu ^{\prime \prime }$
under $\left( \succ ^{O}\cap \succ ^{A}\right) $. Let $s=\mu ^{\prime \prime
}\left( j\right) $. Then, $sP_{i}\mu ^{\prime \prime }\left( i\right) $ and $%
\left( i,j\right) \in \left( \succ _{s}^{O}\cap \succ _{s}^{A}\right) $.
Since $\mu ^{\prime \prime }\left( i\right) R_{i}\mu ^{\prime }\left(
i\right) $ and $\mu ^{\prime \prime }\left( j\right) =\mu ^{\prime }\left(
j\right) =s$, we have $sP_{i}\mu ^{\prime }\left( i\right) $. Thus, $i$ also
has justified envy toward $j$ at $\mu ^{\prime }$ under $\left( \succ
^{O}\cap \succ ^{A}\right) $, contradicting the fairness of $\mu ^{\prime }$
under $\left( \succ ^{O}\cap \succ ^{A}\right) $. Hence $\mu ^{\prime \prime
}$ is fair under $\left( \succ ^{O}\cap \succ ^{A}\right) $.

Let $\mu _{0}=\mu ^{\prime \prime }$. Starting from $\mu _{0}$, consider the
following algorithm.

\begin{description}
\item[\textbf{Step }$k$] For school $s$, let 
\begin{equation*}
I_{k}\left( s\right) =\left\{ \left. i\in I\text{ }\right\vert \text{ }%
sP_{i}\mu _{k-1}\left( i\right) \right\} \text{.}
\end{equation*}%
If there is no school $s$ such that $\left\vert \mu _{k-1}\left( s\right)
\right\vert <q_{s}$ and $I_{k}\left( s\right) \neq \emptyset $, then this
algorithm terminates. Otherwise, then choose such a school $s$ and $i^{\ast
}\in I_{k}\left( s\right) $ such that there is no $i\in I_{k}\left( s\right) 
$ satisfying $\left( i,i^{\ast }\right) \in \left( \succ ^{O}\cap \succ
^{A}\right) $. By Lemma 1, such a student exists. Define $\mu _{k}$ by $\mu
_{k}\left( i^{\ast }\right) =s$ and $\mu _{k}\left( i\right) =\mu
_{k-1}\left( i\right) $ for all $i\in I\setminus \left\{ i^{\ast }\right\} $.
\end{description}

The algorithm terminates after finitely many steps because, whenever a
student is selected, that student is assigned to a strictly

preferred school, and each student has only finitely many possible
assignments. Let $K$ denote the terminal step.

By construction, the following three properties hold. First, for all $%
k=1,\cdots ,K$, $\mu _{k}$ Pareto dominates $\mu _{k-1}$ for $I$. Second,
for all $k=1,\cdots ,K$, if $\mu _{k-1}$ is individually rational and fair
under $\left( \succ ^{O}\cap \succ ^{A}\right) $, then so is $\mu _{k}$.
Individual rationality follows from $i^{\ast }\in I_{k}\left( s\right) $. To
verify fairness, note that the only new instance of justified envy that
could be created at Step $k$ would be toward $i^{\ast }$ at $s$. If some
student $i$ had justified envy toward $i^{\ast }$, then $i^{\ast }\in
I_{k}\left( s\right) $ and $\left( i,i^{\ast }\right) \in \left( \succ
^{O}\cap \succ ^{A}\right) ,$ contradicting the choice of $i^{\ast }$.
Third, $\mu _{K}$ is non-wasteful by the termination condition.

Since $\mu _{0}=\mu ^{\prime \prime }$ is individually rational and fair
under $\left( \succ ^{O}\cap \succ ^{A}\right) $, it follows that $\mu _{K}$
is individually rational, non-wasteful, and fair under $\left( \succ
^{O}\cap \succ ^{A}\right) $. By Proposition 1, $\mu _{K}$ is weakly stable
under $\left( \succ ^{O},\succ ^{A}\right) $. Moreover, since $\mu _{K}$
weakly Pareto dominates $\mu ^{\prime \prime }$ for $I$, and $\mu ^{\prime
\prime }$ Pareto dominates $\mu ^{\ast }$ for $I^{\prime }$, $\mu _{K}$ also
Pareto dominates $\mu ^{\ast }$ for $I^{\prime }$. This contradicts the $%
I^{\prime }$-optimality of $\mu ^{\ast }$.

\subsection*{A.2 Formal description of EADA algorithms}

To prove Proposition 3 and Theorem 1, we first provide a formal description
of the EADA algorithms that are briefly introduced in Section 3.

We first let $G=\left( I,S,P,\succ ,q\right) $ be a school choice problem
where each school has only one partial order priority $\succ _{s}$. To
introduce the formal description of the EADA algorithm, we let a school
choice subproblem be%
\begin{equation*}
\hat{G}=\left( \hat{I},\hat{S},\hat{P},\hat{\succ},\hat{q}\right) \in
2^{I}\times 2^{S}\times \mathbb{P}^{\left\vert I\right\vert }\times \mathcal{%
L}^{\left\vert S\right\vert }\times \mathbb{N}^{^{\left\vert S\right\vert }}
\end{equation*}%
where $\hat{\succ}_{s}$ is a linear order for all $s\in S$. Let $DA\left( 
\hat{G}\right) $ be the outcome of the SPDA algorithm for $\hat{G}$.

A school $s\in \hat{S}$ is said to be \textbf{underdemanded} at a matching $%
\mu $ with $\hat{G}$ if $\mu \left( i\right) \hat{R}_{i}s$ for all $i\in 
\hat{I}$. Then, we have the following result. Note that a school $s$ is
underdemanded at $DA\left( \hat{G}\right) $ with $\hat{G}$\ if and only if $%
s $ never rejects in any steps of the SPDA algorithm for $\hat{G}$.

\begin{description}
\item[Round $0$] Let $\succ ^{\ast }=\left( \succ _{s}^{\ast }\right) _{s\in
S}\in \mathcal{L}^{\left\vert S\right\vert }$ be such that $\succ _{s}^{\ast
}$ is a linear order extension of $\succ _{s}$ for all $s\in S$.

\item[Round $1$] Let $I^{1}=I,$ $S^{1}=S$ and $P^{1}=P$. Run the SPDA for $%
G^{1}=\left( I^{1},S^{1},P^{1},\succ ^{\ast },q\right) $. Let $U^{1}$ be the
set of underdemanded schools at $DA\left( G^{1}\right) $; that is, each $%
s\in U^{1}$ never rejects any student throughout the SPDA of this Round.
Moreover, let 
\begin{equation*}
E^{1}=\bigcup\nolimits_{s\in U^{1}\cup \left\{ \emptyset \right\} }DA\left(
G^{1}\right) \left( s\right) ,
\end{equation*}%
which is the set of students who are matched to an underdemanded school or
unmatched at $DA\left( G^{1}\right) $.

\item[Round $k$] Let $I^{k}=I^{k-1}\setminus E^{k-1}$ and $%
S^{k}=S^{k-1}\setminus U^{k-1}$. For all $i\in I^{k},$ let 
\begin{equation*}
Z_{i}^{k}=\left\{ s\in S^{k}\text{ }\left\vert \text{ }sP_{i}^{k-1}\emptyset
,\text{ }\left( j,i\right) \in \succ _{s}\text{ and }sP_{j}DA\left(
G^{k-1}\right) \left( j\right) ,\text{ for some }j\in E^{k-1}\right.
\right\} \text{,}
\end{equation*}%
and let $P_{i}^{k}$ be a linear order such that for all $s\in Z_{i}^{k}$, $%
\emptyset P_{i}^{k}s$ and for all $s^{\prime },s^{\prime \prime }\in \left(
S\cup \left\{ \emptyset \right\} \right) \setminus Z_{i}^{k},$ $s^{\prime
}P_{i}^{k-1}s^{\prime \prime }$ implies $s^{\prime }P_{i}^{k}s^{\prime
\prime }$. Run the SPDA for $G^{k}=\left( I^{k},S^{k},P^{k},\succ ^{\ast
},q\right) $. Let $U^{k}$ be the set of underdemanded schools at $DA\left(
G^{k}\right) $; that is, each $s\in U^{k}$ never rejects any student
throughout the SPDA of this Round. Moreover, let 
\begin{equation*}
E^{k}=\bigcup\nolimits_{s\in U^{k}\cup \left\{ \emptyset \right\} }DA\left(
G^{k}\right) \left( s\right) ,
\end{equation*}%
which is the set of students who are matched to an underdemanded school (or
unmatched) at $DA\left( G^{k}\right) $.\footnote{%
In each Round $k,$ $E^{k}$ is not empty. First, if no student is rejected,
then all schools are underdemanded and $E^{k}=I^{k}$. Otherwise, then we can
let $i$ be the student who is rejected last by a school. Then, $i$ must be
in $E^{k}$.}\newline
\end{description}

For notational convenience, let $\kappa :I\cup S\rightarrow \mathbb{N}$
satisfy $\kappa \left( i\right) =k$ for all $i\in E^{k}$ and $\kappa \left(
s\right) =k$ for all $s\in U^{k}$; that is, $\kappa \left( i\right) $ and $%
\kappa \left( s\right) $ represent the Rounds in which $i$ and $s$ are
eliminated, respectively. Moreover, let $\mathcal{E}^{k}\mathcal{=}%
\bigcup\nolimits_{k^{\prime }=1}^{k}E^{k^{\prime }},$ which is the set of
students eliminated from Round $1$ to Round $k$.

We let $\hat{\mu}^{k}$ be the matching obtained after Round $k$; that is, 
\begin{eqnarray*}
\hat{\mu}^{k}\left( i\right) &=&DA\left( G^{\kappa \left( i\right) }\right)
\left( i\right) \text{ for all }i\in I\setminus I^{k}\text{ }(\kappa \left(
i\right) <k),\text{ } \\
\hat{\mu}^{k}\left( i\right) &=&DA\left( G^{k}\right) \left( i\right) \text{
for all }i\in I^{k},\text{ }(\kappa \left( i\right) \geq k).
\end{eqnarray*}%
That is, for each student who is remaining in Round $k$, $\hat{\mu}^{k}$ is
the SPDA result in Round $k$, and for each student who is not remaining, it
is the SPDA result in the round when the student is eliminated. In the
terminal step $K$, $S=\bigcup\nolimits_{k=1}^{K}U^{k}$ and the resulting
matching is $\hat{\mu}^{K}$.

Kitahara and Okumura (2026b) show the following two results:

\begin{lemma}
(Kitahara and Okumura 2026, Lemma 1) For all $i\in I$ and all $k=2,3,\cdots
,K,$ $\hat{\mu}^{k}\left( i\right) R_{i}\hat{\mu}^{k-1}\left( i\right) $.
\end{lemma}

\begin{lemma}
(Kitahara and Okumura 2026, Lemma 5) If $\succ \in \mathcal{P}^{\left\vert
S\right\vert }$, then $\hat{\mu}^{k}$ is stable under $\succ $ for all $%
k=1,\cdots ,K$.
\end{lemma}

\subsection*{A.3 Proof of Proposition 3}

First, by the construction of the EADA mechanism in A.2, $\mu ^{\ast \ast }$
is obtained in Round 1; that is, $\hat{\mu}^{1}=\mu ^{\ast \ast }$. By Lemma
5, $\mu ^{\ast }$ weakly Pareto dominates $\mu ^{\ast \ast },$ because $\mu
^{\ast }$ is the result of the EADA mechanism.

Next, suppose that a student $i$ has justified envy toward $j$ at $\mu
^{\ast }$ under $\succ ^{\ast }$. Since $\mu ^{\ast \ast }$ has no justified
envy under $\succ ^{\ast }$, $\mu ^{\ast }\neq \mu ^{\ast \ast }$. By Lemma
5, $\mu ^{\ast }$ Pareto dominates $\mu ^{\ast \ast }$ for\textbf{\ }$I$.
Finally, let $s=\mu ^{\ast }\left( j\right) $. Since $\mu ^{\ast }$ is
weakly stable under $\left( \succ ^{O},\succ ^{A}\right) $, it cannot be the
case that both $i\succ _{s}^{O}j$ and $i\succ _{s}^{A}j$. On the other hand,
because $\succ ^{\ast }$ is a linear-order extension profile of $\left(
\succ ^{O}\cap \succ ^{A}\right) $, $i\succ _{s}^{\ast }j$ rules out $j\succ
_{s}^{O}i$ and $j\succ _{s}^{A}i$, simultaneously. Hence, exactly one of the
two original priority orders ranks $i$ above $j$, and therefore either $%
j\succ _{s}^{O}i$ or $j\succ _{s}^{A}i$.

\subsection*{A.4 Proof of Theorem 1}

Now we prove Theorem 1. Let $\succ ^{O},\succ ^{A}\in \mathcal{L}%
^{\left\vert S\right\vert }$ with $\succ _{s}^{O}\neq \succ _{s}^{A}$ for
some $s\in S$. For notational simplicity, let 
\begin{equation*}
\phi ^{\ast }\left( \succ ^{O},\succ ^{A}\right) =EA\left( \left( \succ
^{O}\cap \succ ^{A}\right) ,\succ ^{A}\right) =\mu ^{\ast }.
\end{equation*}%
Moreover, fix an arbitrary set $I^{\prime }\supseteq I\left( \succ
^{O},\succ ^{A}\right) $.

We show that $\mu ^{\ast }$ is an $I^{\prime }$\textbf{-}optimally weakly
stable\textbf{\ }under $\left( \succ ^{O},\succ ^{A}\right) .$ To this end,
we first establish the following result.

\begin{claim}
Consider the EADA mechanism with $\left( \succ ^{O},\succ ^{A}\right) ,$
which yields $EA\left( \left( \succ ^{O}\cap \succ ^{A}\right) ,\succ
^{A}\right) $. Fix any $k=1,\cdots ,K$. Suppose that $\nu $ is a stable
matching under $\left( \succ ^{O}\cap \succ ^{A}\right) $ that Pareto
dominates $\hat{\mu}^{k}$ for $I^{\prime }$ and $\hat{\mu}^{k}\left(
i^{\prime }\right) R_{i^{\prime }}\nu \left( i^{\prime }\right) $ for all $%
i^{\prime }\in \mathcal{E}^{k-1}$, where $\mathcal{E}^{0}=\emptyset $. Then, 
$\hat{\mu}^{k}\left( j\right) R_{j}\nu \left( j\right) $ for all $j\in 
\mathcal{E}^{k}(=\mathcal{E}^{k-1}\cup E^{k})$.
\end{claim}

\textbf{Proof.} Suppose not; that is, for some $k=1,\cdots ,K$, there is $%
i\in E^{k}$ such that $\nu \left( i\right) P_{i}\hat{\mu}^{k}\left( i\right) 
$. Let $\kappa \left( i\right) =k$. Since $i\in E^{k}$, either $\hat{\mu}%
^{k}\left( i\right) \in U^{k}$ or $\hat{\mu}^{k}\left( i\right) =\emptyset $.

First, we show that there is a student $i_{d}\in I\setminus I^{\prime }$
such that $\hat{\mu}^{k}\left( i_{d}\right) P_{i_{d}}\nu \left( i_{d}\right) 
$ and a student $j_{d}$ such that $\hat{\mu}^{k}\left( i_{d}\right) =\nu
\left( j_{d}\right) $ and $\nu \left( j_{d}\right) P_{j_{d}}\hat{\mu}%
^{k}\left( j_{d}\right) $. Define 
\begin{equation*}
S^{\ast }=\left\{ s\in S\text{ }\left\vert \text{ }s=\nu \left( i^{\prime
}\right) \text{ and }\nu \left( i^{\prime }\right) P_{i^{\prime }}\hat{\mu}%
^{k}\left( i^{\prime }\right) \text{ for some }i^{\prime }\in I\right.
\right\} \text{.}
\end{equation*}%
Then, by the assumption above, $\nu \left( i\right) \in S^{\ast }$.

By contrast, $\hat{\mu}^{k}\left( i\right) \in \left( S\cup \left\{
\emptyset \right\} \right) \setminus S^{\ast }$. We show this fact. Since $%
i\in E^{k}$, either $\hat{\mu}^{k}\left( i^{\prime }\right) =\emptyset $ or $%
\hat{\mu}^{k}\left( i^{\prime }\right) \in U^{k}$. Since the result is
immediate in the former case, suppose $\hat{\mu}^{k}\left( i^{\prime
}\right) =s\in U^{k}$. Moreover, toward a contradiction, we assume $s\in
S^{\ast }$. Then there exists a student $i^{\prime }\in I$ such that 
\begin{equation*}
\nu \left( i^{\prime }\right) =sP_{i^{\prime }}\hat{\mu}^{k}\left( i^{\prime
}\right) .
\end{equation*}%
Since $s$ is underdemanded in Round $k$, no student remaining in that round
prefers $s$ to her assignment under $\hat{\mu}^{k}$. Since $i^{\prime }$
must have been eliminated before Round $k$, $i^{\prime }\in \mathcal{E}%
^{k-1} $. However, $\nu \left( i^{\prime }\right) P_{i^{\prime }}\hat{\mu}%
^{k}\left( i^{\prime }\right) $ contradicts the assumption that $\hat{\mu}%
^{k}\left( i^{\prime }\right) R_{i^{\prime }}\nu \left( i^{\prime }\right) $
for every $i^{\prime }\in \mathcal{E}^{k-1}$. Therefore, $\hat{\mu}%
^{k}\left( i^{\prime }\right) \notin S^{\ast }$.

Let 
\begin{eqnarray*}
I_{1} &=&\left\{ i^{\prime }\in I\text{ }\left\vert \text{ }\nu \left(
i^{\prime }\right) \in S^{\ast }\text{ and }\hat{\mu}^{k}\left( i^{\prime
}\right) \in \left( S\cup \left\{ \emptyset \right\} \right) \setminus
S^{\ast }\right. \right\} , \\
I_{2} &=&\left\{ i^{\prime }\in I\text{ }\left\vert \text{ }\nu \left(
i^{\prime }\right) \in \left( S\cup \left\{ \emptyset \right\} \right)
\setminus S^{\ast }\text{ and }\hat{\mu}^{k}\left( i^{\prime }\right) \in
S^{\ast }\right. \right\} .
\end{eqnarray*}%
Then, $i\in I_{1}\setminus I_{2}$. We show $I_{2}\neq \emptyset $. Suppose
not; that is, $I_{2}=\emptyset $. Since $I_{1}\neq \emptyset $ and $%
I_{2}=\emptyset $, there is $s\in S^{\ast }$ such that $\left\vert \hat{\mu}%
^{k}\left( s\right) \right\vert <\left\vert \nu \left( s\right) \right\vert
\leq q_{s}$. Since $s\in S^{\ast }$, there is $i^{\prime }$ such that $s=\nu
\left( i^{\prime }\right) P_{i^{\prime }}\hat{\mu}^{k}\left( i^{\prime
}\right) $. Moreover, by Lemma 6, $\hat{\mu}^{k}$ is stable under $\left(
\succ ^{O}\cap \succ ^{A}\right) $. However, $\left\vert \hat{\mu}^{k}\left(
s\right) \right\vert <q_{s}$ and $sP_{i^{\prime }}\hat{\mu}^{k}\left(
i^{\prime }\right) $ contradict the non-wastefulness of $\hat{\mu}^{k}$.
Hence $I_{2}\neq \emptyset $.

Then, we show that any student $i_{d}\in I_{2}$ must belong to $I\setminus
I^{\prime }$. Since $\nu \left( i_{d}\right) \notin S^{\ast },$ $\hat{\mu}%
^{k}\left( i_{d}\right) R_{i_{d}}\nu \left( i_{d}\right) $. Moreover, $\nu
\left( i_{d}\right) \notin S^{\ast }$ and $\hat{\mu}^{k}\left( i_{d}\right)
\in S^{\ast }$ imply that $\nu \left( i_{d}\right) \neq \hat{\mu}^{k}\left(
i_{d}\right) $. Therefore, $\hat{\mu}^{k}\left( i_{d}\right) P_{i_{d}}\nu
\left( i_{d}\right) $. Since $\nu $ Pareto dominates $\hat{\mu}^{k}$ for $%
I^{\prime }$, it follows that $i_{d}\in I\setminus I^{\prime }$. Further,
since $\hat{\mu}^{k}\left( i_{d}\right) \in S^{\ast }$, there is $j_{d}\neq
i_{d}$ such that $s=\nu \left( j_{d}\right) $ and $\nu \left( j_{d}\right)
P_{j_{d}}\hat{\mu}^{k}\left( j_{d}\right) $. We let $\hat{\mu}^{k}\left(
i_{d}\right) =s$.

Second, we show $j_{d}\succ _{s}^{A}i_{d}$. Since $\nu \left( j_{d}\right) =%
\hat{\mu}^{k}\left( i_{d}\right) =s$, $sP_{i_{d}}\nu \left( i_{d}\right) $
and $\nu $ is stable under $\left( \succ ^{O}\cap \succ ^{A}\right) $, $%
\left( i_{d},j_{d}\right) \notin \left( \succ _{s}^{O}\cap \succ
_{s}^{A}\right) $. That is, it cannot be the case that both $i_{d}\succ
_{s}^{O}j_{d}$ and $i_{d}\succ _{s}^{A}j_{d}$. Since $i_{d}\in I\setminus
I^{\prime }$, and $I^{\prime }\supseteq I\left( \succ ^{O},\succ ^{A}\right) 
$, $i_{d}$ is not an improved student; that is, $i_{d}\notin I\left( \succ
^{O},\succ ^{A}\right) $. Hence it cannot be the case that $j_{d}\succ
_{s}^{O}i_{d}$ and $i_{d}\succ _{s}^{A}j_{d}$. Therefore, $j_{d}\succ
_{s}^{A}i_{d}$.

Third, we show $\kappa \left( s\right) >k;$ that is, $s$ and $i_{d}$ have
not yet been eliminated by round $k$. Suppose not; that is, $\kappa \left(
s\right) \leq k$. Since $\nu \left( j_{d}\right) =s$ and $sP_{j_{d}}\hat{\mu}%
^{k}\left( j_{d}\right) $, Lemma 5 implies $\nu \left( j_{d}\right) P_{j_{d}}%
\hat{\mu}^{\kappa \left( s\right) }\left( j_{d}\right) $. Since $s$ is
underdemanded at $\kappa \left( s\right) $, it follows that $s$ must have
been removed from the preference list of $j_{d}$ at that round; that is, $%
\emptyset P_{j_{d}}^{\kappa \left( s\right) }s$ Therefore, there is $%
j^{\prime }\in I$ such that $\kappa \left( j^{\prime }\right) <\kappa \left(
s\right) $, $sP_{j^{\prime }}\hat{\mu}^{k}\left( j^{\prime }\right) $ and $%
\left( j^{\prime },j_{d}\right) \in \left( \succ _{s}^{O}\cap \succ
_{s}^{A}\right) $. Since $\nu $ is stable under $\left( \succ ^{O}\cap \succ
^{A}\right) ,$ $\nu \left( j^{\prime }\right) R_{j^{\prime }}s$ and thus $%
\nu \left( j^{\prime }\right) P_{j^{\prime }}\hat{\mu}^{k}\left( j^{\prime
}\right) $. Moreover, since $\kappa \left( j^{\prime }\right) <\kappa \left(
s\right) \leq k$, $j^{\prime }\in \mathcal{E}^{k-1}$. However, these facts
contradict $\hat{\mu}^{k}\left( i^{\prime }\right) R_{i^{\prime }}\nu \left(
i^{\prime }\right) $ for all $i^{\prime }\in \mathcal{E}^{k-1}$. Therefore, $%
s$ and $i_{d}$ have not yet been eliminated at Round $k$ yet; that is, $%
\kappa \left( s\right) >k$.

Moreover, since $s=\nu \left( j_{d}\right) $ and $sP_{j_{d}}\hat{\mu}%
^{k}\left( j_{d}\right) $, $j_{d}\notin \mathcal{E}^{k-1}$, and hence $%
\kappa \left( j_{d}\right) \geq k$. Therefore, in the SPDA of Round $k$, $%
i_{d}$ is accepted by $s$ whereas $j_{d}$ is not. That is, either $j_{d}$ is
rejected by $s$ or $\emptyset P_{j_{d}}^{k}s$. Since $j_{d}\succ
_{s}^{A}i_{d}$, the former cannot occur, and thus $\emptyset P_{j_{d}}^{k}s.$
Hence, there exists $j^{\prime }$ such that $\kappa \left( j^{\prime
}\right) <k$, $\left( j^{\prime },j_{d}\right) \in \left( \succ _{s}^{O}\cap
\succ _{s}^{A}\right) $, and $sP_{j^{\prime }}\hat{\mu}^{k}\left( j^{\prime
}\right) $. Since $\nu $ is stable under $\left( \succ ^{O}\cap \succ
^{A}\right) ,$ 
\begin{equation*}
\nu \left( j^{\prime }\right) R_{j^{\prime }}sP_{j^{\prime }}\hat{\mu}%
^{k}\left( j^{\prime }\right) .
\end{equation*}%
However, since $j^{\prime }\in \mathcal{E}^{k-1}$, this contradicts the
assumption that $\hat{\mu}^{k}\left( i^{\prime }\right) R_{i^{\prime }}\nu
\left( i^{\prime }\right) $ for all $i^{\prime }\in \mathcal{E}^{k-1}$. This
contradiction completes the proof of the claim. \textbf{Q.E.D.}\newline

We now prove Theorem 1. Suppose not; that is, there is a weakly stable
matching under $\left( \succ ^{O},\succ ^{A}\right) $ denoted by $\nu $ such
that $\nu $ Pareto dominates $\mu ^{\ast }$ for $I^{\prime }$. By
Proposition 1, $\nu $ is a stable matching under $\left( \succ ^{O}\cap
\succ ^{A}\right) $. By Lemma 5, $\nu $ also Pareto dominates $\hat{\mu}^{k}$
for $I^{\prime }$ for all $k=1,\cdots ,K$. By Claim 1, $\mathcal{E}%
^{0}=\emptyset $ implies that $\hat{\mu}^{1}\left( j\right) R_{j}\nu \left(
j\right) $ for all $j\in \mathcal{E}^{1}\left( =E^{1}\right) $. Since $\hat{%
\mu}^{1}\left( j\right) =\hat{\mu}^{2}\left( j\right) $ for all $j\in 
\mathcal{E}^{1}$ and Claim 1 is satisfied, $\hat{\mu}^{2}\left( j\right)
R_{j}\nu \left( j\right) $ for all $j\in \mathcal{E}^{2}\left( =E^{2}\cup
E^{1}\right) $. Likewise, we have $\mu ^{\ast }\left( j\right) =\hat{\mu}%
^{K}\left( j\right) R_{j}\nu \left( j\right) $ for all $j\in \mathcal{E}%
^{K}=I$, but this contradicts that$\ \nu $ Pareto dominates $\mu ^{\ast }$
for $I^{\prime }$.

\subsection*{A.5 Proof of Theorem 2}

Suppose that $\succ ^{A}$ represents a stronger improvement of $\succ ^{O}$
than $\succ ^{A\prime }$, and fix any $I^{\prime }\supseteq I\left( \succ
^{O},\succ ^{A}\right) $.

We first show that every matching that is weakly stable under $\left( \succ
^{O},\succ ^{A\prime }\right) $ is also weakly stable under $\left( \succ
^{O},\succ ^{A}\right) $. Let $\mu $ be any weakly stable matching under $%
\left( \succ ^{O},\succ ^{A\prime }\right) $. By Proposition 1, $\mu $ is
stable under $\left( \succ _{s}^{O}\cap \succ _{s}^{A\prime }\right) $. By
Lemma 2, $\mu $ is also stable under $\left( \succ _{s}^{O}\cap \succ
_{s}^{A}\right) $. Hence, again by Proposition 1, $\mu $ is weakly stable
under $\left( \succ ^{O},\succ ^{A}\right) $.

Thus, $\phi ^{\ast }\left( \succ ^{O},\succ ^{A\prime }\right) $ is weakly
stable under $\left( \succ ^{O},\succ ^{A}\right) $. By Theorem 1, since $%
I^{\prime }\supseteq I\left( \succ ^{O},\succ ^{A}\right) $, $\phi ^{\ast
}\left( \succ ^{O},\succ ^{A}\right) $ is not Pareto dominated for $%
I^{\prime }$ by any weakly stable matching under $\left( \succ ^{O},\succ
^{A}\right) $. Therefore, $\phi ^{\ast }\left( \succ ^{O},\succ ^{A}\right) $
is not Pareto dominated for $I^{\prime }$ by $\phi ^{\ast }\left( \succ
^{O},\succ ^{A\prime }\right) $. Therefore, $\phi ^{\ast }$ is responsive to
improvements.

\subsection*{A.6 Another notion of responsiveness}

We show how our notion of responsiveness is related to that introduced by
Jiao and Shen (2021) and Jiao et al. (2022).

Jiao and Shen (2021) (as well as Jiao and Tian (2018) and Jiao et al.
(2022)) consider the following notion of improvement, which differs from
ours. For $I^{\prime }\subseteq I$, we say that $\succ $ is an $I^{\prime }$-%
\textbf{improvement} of $\succ ^{\prime }$ if, for every $s\in S$, (JS1) $%
i\succ _{s}^{\prime }j$ and $i\in I^{\prime }$ imply $i\succ _{s}j$, and
(JS2) $i\succ _{s}^{\prime }j$ and $i,j\in I\setminus I^{\prime }$ imply $%
i\succ _{s}j$.\footnote{%
The notion of improvement introduced here is slightly more general than that
introduced by Jiao and Tian (2018).} Note that $\succ $ is an $I^{\prime }$%
-improvement of itself for any $I^{\prime }\subseteq I$.

We first establish the relationship between their notion of improvement and
ours.

\begin{remark}
Let $I^{\prime }\subseteq I$ and $\succ ^{O},\succ ^{A},\succ ^{A\prime }$
be such that (i) both $\succ ^{A}$ and $\succ ^{A\prime }$ are $I^{\prime }$%
-improvements of $\succ ^{O}$ for $I^{\prime }\subseteq I,$ (ii) $\succ ^{A}$
is an $I^{\prime }$-improvement of $\succ ^{A\prime }$\textbf{\ }for $%
I^{\prime }\subseteq I$. Then, $\succ ^{A}$ represents a stronger
improvement of $\succ ^{O}$ than $\succ ^{A\prime }$.
\end{remark}

\textbf{Proof. }Suppose that $i\succ _{s}^{O}j$ and $j\succ _{s}^{A\prime }i$%
. Since $\succ ^{A\prime }$ is an $I^{\prime }$-improvement of $\succ ^{O}$
for $I^{\prime }\subseteq I$, $j\in I^{\prime }$ and $i\in I\setminus
I^{\prime }$. By (JS1) and $j\in I^{\prime }$, $j\succ _{s}^{A\prime }i$
implies $j\succ _{s}^{A}i$. Thus, $\succ ^{A}$ represents a stronger
improvement of $\succ ^{O}$ than $\succ ^{A\prime }$. \textbf{Q.E.D.}\newline

\begin{remark}
Let $\succ ^{O},\succ ^{A},\succ ^{A\prime }$ be such that (i) both $\succ
^{A}$ and $\succ ^{A\prime }$ are $I^{\prime }$-improvements of $\succ ^{O}$
for $I^{\prime }\subseteq I,$ (ii) $\succ ^{A}$ is an $I^{\prime }$%
-improvement of $\succ ^{A\prime }$\textbf{\ }for $I^{\prime }\subseteq I$.
Then, $\phi ^{\ast }\left( \succ ^{O},\succ ^{A}\right) $ is not Pareto
dominated by $\phi ^{\ast }\left( \succ ^{O},\succ ^{A\prime }\right) $ for $%
I^{\prime }$.
\end{remark}

\textbf{Proof.} Since $\succ ^{A}$ is an $I^{\prime }$-improvement of $\succ
^{O}$, (JS1) and (JS2) imply that no student outside $I^{\prime }$ can be an
improved student under $\left( \succ ^{O},\succ ^{A}\right) $. Hence, $%
I\left( \succ ^{O},\succ ^{A}\right) \subseteq I^{\prime }$. By Theorem 2, $%
\phi ^{\ast }\left( \succ ^{O},\succ ^{A}\right) $ is not Pareto dominated
by $\phi ^{\ast }\left( \succ ^{O},\succ ^{A\prime }\right) $ for $I^{\prime
}$. \textbf{Q.E.D.}\newline

In Jiao and Shen (2021), $I^{\prime }$ denotes the set of minority students.
Given two priority profiles $\succ ^{A}$ and $\succ ^{A\prime }$ such that $%
\succ ^{A}$ is an $I^{\prime }$-improvement of $\succ ^{A\prime }$, a
mechanism is said to be responsive to priority-based affirmative action if
the matching induced by $\succ ^{A\prime }$ does not Pareto dominate the
matching induced by $\succ ^{A}$ for $I^{\prime }$. The preceding remarks
show that our mechanism $\phi ^{\ast }$ satisfies this responsiveness
whenever both $\succ ^{A}$ and $\succ ^{A\prime }$ are $I^{\prime }$%
-improvements of $\succ ^{O}$. Indeed, since $I\left( \succ ^{O},\succ
^{A}\right) \subseteq I^{\prime }$, our responsiveness requirement applies
to the entire minority group $I^{\prime }$, including minority students
whose priorities are not improved.

In Jiao and Shen (2021), $I^{\prime }$ denotes the set of minority students.
Given two priority profiles $\succ ^{A}$ and $\succ ^{A\prime }$ such that $%
\succ ^{A}$ is an $I^{\prime }$-improvement of $\succ ^{A\prime }$, a
mechanism is said to be responsive to priority-based affirmative action if
the matching induced by $\succ ^{A\prime }$ does not Pareto dominate that
induced by $\succ ^{A}$ for $I^{\prime }$. Although Jiao and Shen (2021) do
not explicitly specify an original priority profile, their responsiveness
requirement can be compared with ours by fixing a baseline profile $\succ
^{O}$. The preceding remarks show that our mechanism $\phi ^{\ast }$
satisfies their responsiveness requirement whenever both $\succ ^{A}$ and $%
\succ ^{A\prime }$ are $I^{\prime }$-improvements of $\succ ^{O}$. Indeed,
since $I\left( \succ ^{O},\succ ^{A}\right) \subseteq I^{\prime }$, our
responsiveness requirement applies to the entire minority group $I^{\prime }$%
, including minority students whose priorities are not improved relative to $%
\succ ^{O}$.

\end{document}